%% file: main.tex
\documentclass[aps,pra,reprint,superscriptaddress,nofootinbib,longbibliography]{revtex4-2}

\usepackage[T1]{fontenc}
\usepackage{amsmath,amssymb,amsthm,mathtools}
\usepackage{physics}
\usepackage{bm,dsfont}
\usepackage{microtype}
\usepackage{xcolor}
\usepackage{hyperref}
\usepackage{cleveref}
\usepackage{enumitem}
\usepackage{tikz}
\usetikzlibrary{arrows.meta,calc,fit,positioning,decorations.pathreplacing}

\definecolor{encoderblue}{RGB}{50,102,166}
\definecolor{encoderteal}{RGB}{32,133,120}
\definecolor{encoderorange}{RGB}{211,122,42}
\definecolor{softblue}{RGB}{231,239,249}
\definecolor{softteal}{RGB}{229,244,241}
\definecolor{softorange}{RGB}{250,238,225}

\definecolor{linkburgundy}{HTML}{7A3341}
\definecolor{citeteal}{HTML}{2F6B67}
\definecolor{urlblue}{HTML}{355F8A}

\hypersetup{
    colorlinks = true,
    linkcolor  = linkburgundy,
    citecolor  = citeteal,
    urlcolor   = urlblue
}

\newtheorem{theorem}{Theorem}
\newtheorem{lemma}[theorem]{Lemma}
\newtheorem{proposition}[theorem]{Proposition}
\newtheorem{corollary}[theorem]{Corollary}
\newtheorem{definition}[theorem]{Definition}
\newtheorem{remark}[theorem]{Remark}

\crefname{theorem}{Theorem}{Theorems}
\crefname{lemma}{Lemma}{Lemmas}
\crefname{proposition}{Proposition}{Propositions}
\crefname{corollary}{Corollary}{Corollaries}
\crefname{definition}{Definition}{Definitions}
\crefname{remark}{Remark}{Remarks}

\crefname{section}{Section}{Sections}
\crefname{appendix}{Appendix}{Appendices}
\crefname{table}{Table}{Tables}
\crefname{figure}{Figure}{Figures}

\renewcommand{\openone}{\mathds{1}}
\newcommand{\EG}{E_{\mathrm{gen}}}
\newcommand{\ECI}{E_{\mathrm{CR}}}
\newcommand{\EI}{E_{\mathrm I}}

\begin{document}

\title{The Entanglement Content of Quantum Measurement Bases}

\author{Jef Pauwels}
\affiliation{Department of Applied Physics, University of Geneva, 1211 Geneva, Switzerland}
\affiliation{Constructor University, Bremen, Germany}
\author{Otfried G\"uhne}
\affiliation{Naturwissenschaftlich-Technische Fakultät, Universität Siegen, Walter-Flex-Straße 3, 57068 Siegen, Germany}

\begin{abstract}
Bell-state measurements are essential ingredients in 
many protocols for quantum information processing, ranging from quantum 
teleportation and dense coding to entanglement distribution in quantum 
networks. Their power relies on the fact that they are measurements in
an entangled basis of a two-particle system and that the used Bell-state
basis can be generated from a single Bell state by local unitary 
transformations. How can  these measurements be generalized to more 
particles?
We develop a general framework for this state-to-measurement problem: 
We introduce a hierarchy of classes of measurement bases, distinguished 
by the local transformations the parties may use for their generation 
from a single state. This leads to a generalization of the concept of 
maximally entangleable (or weighted hypergraph) states and the 
identification of a novel maximally entangled basis of four qubits, being a candidate for data-hiding tasks or distillation protocols. 
Finally, we prove that not all forms of entanglement can be encoded in an
entire measurement basis.
\end{abstract}

\maketitle

\section{Introduction}\label{sec:introduction}

Entanglement enters quantum theory in two ways: through states and through
measurements.  The entanglement of states has been studied intensively for
three decades, and a rich theory of its classification, transformation, and
quantification is available \cite{Horodecki2009,GuhneToth2009}.  Entangled
measurements---joint measurements whose outcomes project onto entangled
states---are no less fundamental.  The Bell-state measurement powers
teleportation \cite{Bennett1993}, dense coding \cite{Bennett1992}, and
entanglement swapping \cite{Zukowski1993}, which underlies quantum repeaters
\cite{Briegel1998} and quantum-network protocols \cite{Tavakoli2022}.  Yet the
entanglement of measurements has received far less systematic
attention than that of states.  Here we address a basic
structural question: \emph{which types of entanglement can occur uniformly
among the outcomes of a quantum measurement?}

A rank-one projective measurement on $n$ qubits is nothing but an
orthonormal basis, so this is equivalently a question about bases.  The paradigmatic
example is the Bell basis: its four vectors are all maximally entangled,
and they are all related to one another by local unitaries.  In the
bipartite setting this structure is well understood---Werner showed that
maximally entangled bases, orthonormal bases of unitary operators, and tight
teleportation and dense-coding schemes are all equivalent objects
\cite{Werner2001}.  

Maximal entanglement, however, is not the only
interesting possibility.  A canonical example is the \emph{elegant joint
measurement} (EJM) \cite{Gisin2019}: a two-qubit basis whose four outcomes
are \emph{partially} entangled, all to the same degree, with their marginals
pointing to the vertices of a regular tetrahedron on the Bloch sphere.  The
EJM has become a primitive in network nonlocality, where it
generates quantum correlations in the triangle network and violates bilocal
Bell inequalities \cite{Tavakoli2021,Tavakoli2022}; the characteristic EJM
distribution in the triangle network was recently proved nonclassical
\cite{GittonRenner2025}.  Bases of this kind,
whose elements all belong to a single local-unitary orbit, are called
\emph{iso-entangled}.  For two qubits, all iso-entangled projective bases
have recently been classified \cite{DelSanto2024}.  For more qubits,
systematic constructions based on orbits of Pauli subgroups produce
symmetric iso-entangled measurements \cite{Pauwels2026Pauli}, including
the multiqubit EJM and tunable families that generalize it
\cite{PauwelsGisin2026EJM,Pauwels2026Tunable}.
Related constructions interpolate between product and maximally entangled
bases \cite{Karimipour2006,Gheorghiu2010,Rajchel2018} or
combine iso-entanglement with mutual unbiasedness \cite{Czartowski2020}.

All of these works start from the basis: they design bases realizing a
particular entanglement pattern.  In this paper we take the complementary,
state-first point of view:
\begin{quote}
\emph{Given an $n$-qubit pure state $\ket\psi$, is there a complete
orthonormal basis all of whose elements carry exactly the entanglement of
$\ket\psi$?}
\end{quote}
Here ``exactly the same entanglement'' has an unambiguous meaning: 
two states have the same
entanglement type when they are related by local unitaries.  The question is
thus whether the local-unitary orbit of $\ket\psi$ contains $2^n$
mutually orthogonal states---equivalently, whether there exist product
unitaries $U_a$ such that $\{U_a\ket\psi\}_{a=1}^{2^n}$ is an
orthonormal basis.  We call such a family of product unitaries a
\emph{local encoding} of $\ket\psi$, and the resulting basis an
iso-entangled basis with \emph{fiducial state} $\ket\psi$, meaning the
state one starts from.  The name comes from
dense coding: the parties sharing $\ket\psi$ encode one of $2^n$
messages by applying local unitaries, and the message can be read out
perfectly by a single joint measurement precisely when the encoded states
form an orthonormal basis \cite{Bennett1992,Bruss2004,Mozes2005,Wang2008}.
This state-to-basis question was
posed explicitly by Tanaka, Markham, and Murao in 2007, who asked whether
every multiqubit state admits such a local encoding and left the problem
open \cite{Tanaka2007}.

What is known?  The cleanest answer concerns the most restricted encoders.
Suppose each party $j$ encodes one private bit $a_j$ by choosing between
two local unitaries, so that $U_a=\bigotimes_jU_j(a_j)$.  Kruszy\'nska and
Kraus proved that the states encodable in this way are exactly the
\emph{locally maximally entangleable} (LME) states: those that become
equal-weight superpositions of all computational basis states, i.e.,
\emph{flat} states, in some product basis \cite{Kraus2009}.  This class is large,
containing in particular all stabilizer and graph states, but it is far from
everything.  The class also contains hypergraph states
\cite{Qu2013,Rossi2013,Guehne2014} for
which purification protocols and photonic measurement-based processing have
recently been developed \cite{Vandre2023, Huang2024},
and may also be described as weighted hypergraph states.
Already for three
qubits, however, the $W$ state is not LME
\cite{Kraus2009}; recall that GHZ and $W$ represent the two inequivalent
forms of genuine tripartite entanglement \cite{Duer2000}.  

Beyond the LME states, the picture is a collection of constructions.  Miyake
and Briegel generated a basis of $W$-type states from a set of non-commuting Pauli operations \cite{Miyake2005}; the same work that posed the
universality question also gave sufficient constructions for broad
stabilizer-related families and showed that every $W_n$ state is locally
encodable by Pauli operations \cite{Tanaka2007}; and the Pauli-orbit measurements mentioned above provide
systematic multiqubit examples \cite{Pauwels2026Pauli,Pauwels2026Tunable}.
More recently, Pimpel, Renner, and Tavakoli proved that every
bipartite pure state of local dimension $2$, $4$, or $8$ is encodable,
gave a recursive construction of bases for $W_n$, and found convincing
numerical evidence that every state of three qubits or of two qutrits is
encodable \cite{Pimpel2023}.  For four qubits, however, their searches
failed on some states. So they conjectured that local encodability is not
universal, supporting this with a heuristic count of variables and
constraints, and their conjecture has been included in a
list of open problems in entanglement theory
\cite{WeinbrennerInPreparation}.  Interestingly, their explicit four-qubit candidate turns out
to be encodable after all: we give an exact encoding of it in
\cref{app:eta-encoding}.  The conjecture therefore remained open: do the
numerical failures reflect a genuine obstruction, or only the difficulty of
the search?

In this paper we answer this question.  In doing so, we propose an
organizing principle for the entanglement of measurements: encodings are
grouped by the class of local encoders the parties may use, and each class
is studied through the set of entanglement types it can reach.  The first
and most physical axis separating such classes is the \emph{coordination}
among the parties.  Label the basis elements by an $n$-bit string $a$.  In an
independent encoding, as in the Kruszy\'nska--Kraus theorem, the unitary
applied by party $j$ depends only on its own bit $a_j$.  In a
coordinated encoding, each party's unitary may depend on the entire label
$a$.  Coordination refers only to how the classical label is used; every
encoder remains a product unitary, and no entangling operation is ever
applied.  

Between the independent and the fully general case lies a natural
intermediate class, the subject of our first main result: encoding families
that commute pairwise and whose local factors are binary observables, that
is, commuting product reflections.
These are essentially Pauli observables in rotated local bases, as we will see.
In this class the parties may coordinate freely; what is restricted is
the algebraic structure of the encoders, a second axis independent of
coordination. 

\input{figures/encoding_hierarchy}

Our first main result is an exact characterization of this intermediate
class: a state admits a commuting-reflection encoding if and only if it is,
up to local unitaries, a flat superposition in a \emph{graph-state basis},
the joint eigenbasis of a graph-state stabilizer
(\cref{thm:commuting-characterization}).  This is a very natural
generalization of the Kruszy\'nska--Kraus theorem: the flatness condition is
the same, but the reference basis is allowed to be the eigenbasis of any
graph, with the empty graph recovering the product basis and hence exactly
the LME class.  The generalization strictly enlarges the accessible
entanglement types.  We construct an explicit commuting-reflection encoding
of the three-qubit $W$ state (\cref{prop:w-encoding})---a state that
independent encoders provably cannot reach.  

Returning to the
Bell-basis motivation, we consider the four-qubit Higuchi--Sudbery state $M_4$ {\cite{Higuchi2000}} which can be considered to be
the maximally entangled four-qubit state {\cite{Derksen2017, Gour2010}}, e.g., in terms of
the geometric measure of entanglement. For this state
we construct a natural four-qubit analogue of the Bell basis: a complete
iso-entangled basis whose sixteen elements globally maximize both the
geometric measure of entanglement and the average two-versus-two linear
entropy.  A global measurement distinguishes a uniformly chosen basis state
perfectly, whereas every fully separable measurement---and hence every LOCC
measurement---has average success probability at most $2/9$.  Yet $M_4$ admits no commuting-reflection
encoding, so it exactly separates commuting-reflection encoders from general
local encoders (\cref{prop:m4-separator}).

Our second main result answers this question from 2007
\cite{Tanaka2007} and proves the later four-qubit conjecture
\cite{Pimpel2023}: local encoding is not universal.  Within an $18$-dimensional family of
four-qubit states whose one-qubit marginals are all maximally mixed---the
1-uniform states, the first level of the $k$-uniform hierarchy that
underlies absolutely maximally entangled states and quantum
error-correcting codes \cite{Scott2004,Goyeneche2014}---the
encodable states have at most $15$ independent parameters.  They therefore
cannot fill the family.  We also give an explicit nonencodable state.  Its phases are chosen so that a classical
theorem rules out every possible encoding exactly.  The proof is a dimension comparison, but the
bound does not come from naively counting equations and unknowns, which is
inconclusive here.  Instead it uses a signed-area identity obeyed by every
continuously varying orthonormal basis.  Combined with local encoding and
1-uniformity, this identity forces the signed area of every pair of state
changes within a continuously encodable family to vanish.

Our results are conveniently summarized in terms of the sets of states that
each encoder class can reach.  Write $\EI$, $\ECI$, and $\EG$ for the
$n$-qubit states admitting, respectively, an independent, a
commuting-reflection, and a general local encoding, with a superscript
$(n)$ when the number of qubits matters.  Each set is a union of
local-unitary orbits, and larger encoder classes can only reach more
states, so $\EI\subseteq\ECI\subseteq\EG$.  Our three witnesses turn every
one of these inclusions into a strict one at four qubits,
\begin{equation}\label{eq:strict-four-qubit-hierarchy}
  \EI^{(4)}\subsetneq\ECI^{(4)}\subsetneq\EG^{(4)}
  \subsetneq\mathbb{CP}^{15},
\end{equation}
the last set being all four-qubit pure states.  The three separations are
witnessed by $\ket{W_3}\otimes\ket0$,
by $\ket{M_4}$, and by
$\ket{\Psi_{\mathrm{ex}}}$ (see \cref{cor:explicit-nonencodable}).
\Cref{fig:state-class-map} shows the resulting picture.  Widening the
encoder class genuinely enlarges the entanglement types accessible to projective
measurements, but no amount of encoder freedom makes local encoding universal.  Whether
a given entanglement type can fill a complete measurement is a nontrivial
compatibility condition between the geometry of its local-unitary orbit and
the orthogonality constraints of a basis.  More broadly, the results
suggest studying the entanglement of measurements through the lens of
encoding classes, a perspective we return to in the discussion.

The paper is organized as follows.  \Cref{sec:setting} introduces basic preliminaries.
\Cref{sec:commuting-reflections} proves the commuting-reflection
characterization.  \Cref{sec:w-state} proves strictness using the $W$
state, and \cref{sec:m4-separator} uses $M_4$ to separate
commuting-reflection encoders from general local encoders.
\Cref{sec:nonencodability} proves the nonencodability theorem and gives an
explicit example.  Finally, \cref{sec:discussion} discusses implications
and open problems.

\section{Local-unitary orbits and encoder structures}\label{sec:setting}

We consider $n$ qubits with Hilbert space
$\mathcal H=(\mathbb C^2)^{\otimes n}$ and dimension $N=2^n$.  Pure states
are rays $[\psi]\in\mathbb{CP}^{N-1}$; global phases are never relevant.
The central object of this paper can be read in two equivalent ways.  As a
measurement, it is an orthonormal basis all of whose outcomes have the same
entanglement type.  As an encoding, it is the orbit of a fiducial state
under a finite family of product unitaries.  \Cref{fig:orbit-concept}
illustrates this geometric picture.

\input{figures/orbit_concept}

\begin{definition}[Iso-entangled basis]
An orthonormal basis $\{\ket{\phi_a}\}_{a=1}^N$ is
\emph{iso-entangled} if all of its vectors belong to one local-unitary orbit.
\end{definition}

\begin{definition}[Local encoding]\label{def:local-encoding}
A \emph{local encoding} of a fiducial state $\ket\psi$ is a family of
product unitaries
\[
  U_a=\bigotimes_{j=1}^n U_{a,j},
  \qquad a=1,\ldots,N,
\]
such that
\begin{equation}\label{eq:encoding-condition}
  \langle\psi|U_a^\dagger U_b|\psi\rangle=\delta_{ab}.
\end{equation}
A state that admits such a family is \emph{locally encodable}.
\end{definition}

Equivalently, an encoded basis is a minimum-cardinality complex projective
$1$-design contained in a single local-unitary orbit, since
\[
  \frac1N\sum_{a=1}^N U_a\ket\psi\!\bra\psi U_a^\dagger
  =\frac{\openone}{N}.
\]
In dimension $N$, a projective $1$-design has at least $N$ points, with
equality precisely for an orthonormal basis \cite{RoyScott2007}.

Every locally encoded basis is iso-entangled.  Conversely, anchoring an
iso-entangled basis at any one of its vectors produces a local encoding, so
the two notions describe the same objects.

We use three levels of encoder structure.

\begin{definition}[Encoder structures]\label{def:encoder-structures}
An encoding is
\begin{enumerate}[label=(\roman*)]
  \item \emph{independent} if its labels are bit strings
  $a=(a_1,\ldots,a_n)$ and
  $U_a=\bigotimes_j U_j(a_j)$: party $j$ uses only the private bit $a_j$;
  \item \emph{commuting-reflection} if the $U_a$ commute pairwise and every
  local factor $U_{a,j}$ is a Hermitian unitary involution,
  $U_{a,j}^\dagger=U_{a,j}$ and $U_{a,j}^2=\openone$;
  \item \emph{general} if no restriction beyond product unitarity is imposed.
\end{enumerate}
\end{definition}

An involution is an operator whose square is the identity.  In the
present setting it is also Hermitian and unitary, so it is a $\pm1$-valued
observable.  On a single qubit the trivial cases are $\pm\openone$; every
other example is $VZV^\dagger=\hat{\bm r}\cdot\bm\sigma$, a Pauli observable
in a rotated basis.  In a commuting-reflection encoding it is the complete
product encoders that commute; their individual single-qubit factors need
not.  The parties may therefore coordinate arbitrarily, while the full
product observables share a common eigenbasis.  As we prove in
\cref{sec:commuting-reflections}, every complete encoding family of this
kind can, after a local change of basis, be represented by a stabilizer group
of Pauli operators.  General encodings go beyond this in
two ways: the encoders need not commute (the $W$-basis of Miyake and
Briegel \cite{Miyake2005} and our $M_4$ basis in \cref{app:w3-m4-proofs}
are of this kind), and commuting encoders of higher order, whose local
factors are not Hermitian, are also allowed.

We write $\EI,\ECI,\EG$ for the corresponding sets of state rays.  These
sets are invariant under local unitaries, because conjugating every encoder
preserves both product structure and the stated encoder structure.  At
the level of state classes,
\begin{equation}\label{eq:class-inclusions}
  \EI\subseteq\ECI\subseteq\EG.
\end{equation}
The first inclusion follows from the characterization
below.\footnote{It should not be read as saying that every particular
independent encoder is already a commuting family; the statement is about
the sets of states, not about individual encoding families.}
When the number of qubits matters, we add a superscript, as in
$\EI^{(n)}$, $\ECI^{(n)}$, and $\EG^{(n)}$.

The distinction between independent and coordinated encoding concerns only
which part of the classical label each local factor may use
(\cref{fig:encoder-coordination}). In the general class, the local factor $U_{a,j}$ may be
chosen independently for every pair $(a,j)$; this is the
unrestricted product-unitary model of
Refs.~\cite{Tanaka2007,Pimpel2023}.

\input{figures/encoder_coordination}

\subsection{Independent choices and LME states}

Our starting point is the following theorem.

\begin{theorem}[Kruszy\'nska--Kraus \cite{Kraus2009}]
\label{thm:kraus}
For an $n$-qubit pure state $\ket\psi$, the following are equivalent:
\begin{enumerate}[label=(\alph*)]
  \item $\ket\psi$ admits an independent local encoding;
  \item $\ket\psi$ is local-unitary equivalent to a state
  that is an equal superposition of computational basis states,
  \begin{equation}\label{eq:lme-flat}
  \ket{\psi} \stackrel{LU}{\sim}
    \frac1{\sqrt N}\sum_{x\in\{0,1\}^n}e^{i\alpha_x}\ket x;
  \end{equation}
  \item $\ket\psi$ is locally maximally entangleable.
\end{enumerate}
\end{theorem}

For a bit string $a$, we write
$Z^a=\bigotimes_{j=1}^n Z_j^{a_j}$.
The amplitudes in \eqref{eq:lme-flat} have equal modulus in a product basis.  Applying the
commuting reflections $Z^a$ to \eqref{eq:lme-flat} produces an orthonormal
basis.  Thus every LME state admits a commuting-reflection encoding, which
establishes $\EI\subseteq\ECI$.

\subsection{Graph-state notation}

The Kruszy\'nska--Kraus theorem can be read as a statement about a common
eigenbasis.  After local basis changes, the independent encoders are the
phase flips $Z^a$.  Their common eigenbasis is a product basis, and
orthogonality of the orbit forces the fiducial state to have equal
probabilities in that basis.  Our generalization retains the two ingredients
that make this argument exact---a complete family of commuting
reflections---but drops independence of the local choices.  The common
eigenbasis may then be entangled.  The theorem below shows that it must be a
stabilizer, hence graph-state, basis, and that equal probabilities in this
basis are again necessary and sufficient.

We now introduce the notation needed to make this statement precise.  Graph
states have previously been used to construct multipartite families of
equally entangled bases \cite{Gheorghiu2010}.  Our result is different in
kind: it gives a necessary-and-sufficient condition tied to a specified
algebra of local encoders.

For a simple graph $G$ on $n$ vertices, let
\begin{equation}\label{eq:graph-state}
  \ket G=\prod_{\{j,k\}\in E(G)}\mathrm{CZ}_{jk}\ket+^{\otimes n}.
\end{equation}
Its stabilizer is generated by
\begin{equation}\label{eq:graph-generators}
  K_j=X_j\prod_{k\in\mathcal N(j)}Z_k,
  \qquad j=1,\ldots,n.
\end{equation}
We write
$\mathcal S(G)=\langle K_1,\ldots,K_n\rangle$
for the $2^n$-element stabilizer group of $\ket G$.
Binary sums and dot products below are understood modulo two.
The $N$ states
\begin{equation}\label{eq:graph-basis}
  \ket{G_b}:=Z^b\ket G,
  \qquad b\in\{0,1\}^n,
\end{equation}
form the graph basis, the joint eigenbasis of the stabilizer generators.
We use the standard fact that every maximal qubit stabilizer group (an abelian subgroup of the Pauli group with $2^n$ elements)is
local-Clifford equivalent to a stabilizer group of a graph state
\cite{Gottesman1997,VandenNest2004,Hein2006}.

\section{Commuting-reflection encoders: an exact graph-basis characterization}
\label{sec:commuting-reflections}

We prove the characterization in three steps.  First, pairwise
commutativity removes any continuous freedom in the relative local reflection
axes: a local change of basis can be used to turn every encoder into a tensor product of
Pauli operators.  Second,
$2^n$ distinct commuting Pauli labels necessarily exhaust a maximal
stabilizer group.  Third, orthogonality of its orbit is equivalent to uniform
weight in the joint stabilizer eigenbasis.  The first step is isolated in the
following lemma.

\begin{lemma}[Alignment of commuting product reflections]
\label{lem:reflection-alignment}
Let $\{R_a\}_{a\in\mathcal A}$ be pairwise-commuting product operators
\[
  R_a=\bigotimes_{j=1}^n R_{a,j},
\]
where every $R_{a,j}$ is a single-qubit Hermitian unitary involution.  There
is a local unitary $V=\bigotimes_j V_j$ such that every
$VR_aV^\dagger$ is a tensor product of Pauli operators, up to a real sign.
\end{lemma}

\begin{proof}
Apart from $\pm\openone$, a single-qubit Hermitian involution is a
reflection $\hat{\bm r}\cdot\bm\sigma$ about a Bloch axis
$\hat{\bm r}$.  We first record how two of them can be related.  For
$A=\hat{\bm r}\cdot\bm\sigma$ and $B=\hat{\bm s}\cdot\bm\sigma$,
\[
  AB=(\hat{\bm r}\cdot\hat{\bm s})\openone
  +i(\hat{\bm r}\times\hat{\bm s})\cdot\bm\sigma,
\]
and $BA$ is the same expression with the cross product reversed.  Suppose
$BA=\lambda AB$ for some scalar $\lambda$.  Comparing the two
expressions term by term, the identity part gives
$(\hat{\bm r}\cdot\hat{\bm s})(1-\lambda)=0$ and the Pauli part gives
$(1+\lambda)\,\hat{\bm r}\times\hat{\bm s}=0$.  For unit vectors the dot
and cross products cannot both vanish, so there are only two possibilities:
$\lambda=1$ with parallel axes, or $\lambda=-1$ with perpendicular axes.
In words, two single-qubit reflections either commute and share an axis, or
anticommute and have perpendicular axes.

Now fix two encoders $R_a$ and $R_b$.  The products $R_aR_b$ and
$R_bR_a$ are equal, and both are tensor products of the corresponding
single-qubit factors.  Two such products can agree only if their factors
agree qubit by qubit up to constants whose product is one, so at every
qubit $j$,
\[
  R_{b,j}R_{a,j}=\lambda_jR_{a,j}R_{b,j},
  \qquad \prod_j\lambda_j=1.
\]
By the first paragraph each $\lambda_j=\pm1$, the case of a scalar factor
$\pm\openone$ being trivial.  Hence at every qubit the two axes are
parallel or perpendicular.

Applying this to all pairs of encoders, the axes occurring at a fixed qubit
are pairwise parallel or perpendicular.  Distinct axes at that qubit are
therefore mutually perpendicular, so there are at most three of them, and a
single rotation $V_j$ sends them to the $x,y,z$ axes.  Doing this
independently at every qubit turns every encoder into a tensor product of
Pauli operators, up to a sign.
\end{proof}

\begin{theorem}[Characterization of commuting-reflection encodings]
\label{thm:commuting-characterization}
For an $n$-qubit pure state $\ket\psi$, the following are equivalent:
\begin{enumerate}[label=(\alph*)]
  \item $\ket\psi$ admits a commuting-reflection local encoding;
  \item there is a local unitary $V$ and a graph $G$ such that the
  stabilizer orbit
  \[
    \{\,V^\dagger S V\ket\psi:S\in\mathcal S(G)\,\}
  \]
  is an orthonormal basis;
  \item there is a local unitary $V$, a graph $G$, and phases
  $\{\beta_b\}$ such that
  \begin{equation}\label{eq:flat-graph-basis}
    V\ket\psi
    =\frac1{\sqrt N}
      \sum_{b\in\{0,1\}^n}e^{i\beta_b}\ket{G_b}.
  \end{equation}
\end{enumerate}
\end{theorem}

\input{tables/parallel_characterizations}

\begin{proof}
Assume first \eqref{eq:flat-graph-basis}.  A stabilizer element
$S_a=\prod_jK_j^{a_j}$ acts diagonally in the graph basis:
\[
  S_a\ket{G_b}=(-1)^{a\cdot b}\ket{G_b}.
\]
Consequently,
\[
  \langle\psi|V^\dagger S_a^\dagger S_{a'}V|\psi\rangle
  =\frac1N\sum_b(-1)^{(a\oplus a')\cdot b}
  =\delta_{aa'}.
\]
The conjugated stabilizer elements are commuting product reflections, proving
(c)$\Rightarrow$(b)$\Rightarrow$(a).

Conversely, let $\{R_a\}_{a=1}^N$ be a commuting-reflection encoding.  By
\cref{lem:reflection-alignment} we may assume, after applying a local
unitary $V_0$, that every $R_a$ is a tensor product of Pauli operators,
up to a sign.  No two of them agree up to a sign, since otherwise two
elements of the encoded basis would coincide up to a phase.  We therefore
have $2^n$ commuting Pauli products that are distinct up to signs.

Once signs are discarded, an abelian group of Pauli products has at most
$2^n$ elements; this is the counting behind the stabilizer formalism
\cite{Gottesman1997}.  Our encoders generate such a group and already supply
$2^n$ distinct elements, so they exhaust it.  In particular, the family is
closed under multiplication up to signs, and one of the encoders is the
identity.  Choose $n$ encoders whose products give all $2^n$ elements up
to sign; being $n$ generators of a group of order $2^n$, they are
independent, meaning that no product of a nonempty subset of them is the
identity up to a sign.  Each is Hermitian and squares to $\openone$, so all
their products are as well.  Moreover no such product equals $-\openone$:
that product would be the identity up to a sign, contradicting
independence.  The $n$ chosen encoders therefore generate a stabilizer
group $\mathcal S$ of order $2^n$.

A local Clifford $C$ maps the joint eigenbasis of $\mathcal S$ to the
graph basis of some graph $G$.  Write
\[
  CV_0\ket\psi=\sum_b c_b\ket{G_b}.
\]
Comparing every element of the encoded basis with the one whose encoder is
the identity, orthogonality says precisely that each nonidentity element of
$\mathcal S$ has zero expectation value in the state.  Since
$S_a\ket{G_b}=(-1)^{a\cdot b}\ket{G_b}$, these conditions read
\[
  \sum_b |c_b|^2(-1)^{a\cdot b}=0
  \qquad\text{for every }a\neq0.
\]
Together with $\sum_b|c_b|^2=1$, this is a system of $N$ linear
equations for the $N$ numbers $|c_b|^2$ whose matrix of signs
$(-1)^{a\cdot b}$ is invertible.  Its solution is therefore unique, and
$|c_b|^2=1/N$ solves it.  This is \eqref{eq:flat-graph-basis} with
$V=CV_0$.
\end{proof}

For the empty graph, the graph basis is a product basis.  Hence
\eqref{eq:flat-graph-basis} reduces, after local Hadamards, to the
Kruszy\'nska--Kraus form \eqref{eq:lme-flat}.  The theorem therefore
preserves the equal-modulus condition of the independent case but changes
the reference basis: a nonempty graph replaces the product basis by an
entangled stabilizer basis.  This additional freedom is strictly more powerful, as the
next section shows.

Commuting encoders whose local factors are not binary observables are not
covered by the characterization, and identifying the states they can encode
is one of the open problems we highlight in \cref{sec:discussion}.

The characterization also has a useful ancilla picture, which is how the LME
class is usually presented.  For any table
$\mathcal U=\{U_a\}_{a\in\{0,1\}^n}$, let
$\ket a_A=\bigotimes_j\ket{a_j}_{A_j}$ and control the encoders coherently
on an $n$-qubit label register,
\begin{equation}\label{eq:coherent-control-state}
  \begin{aligned}
  C_{\mathcal U}
    &=\sum_{a\in\{0,1\}^n}\ketbra a a_A\otimes U_a,\\
  \ket{\Omega_{\mathcal U}(\psi)}
    &=C_{\mathcal U}
      \bigl(\ket+_A^{\otimes n}\otimes\ket\psi_S\bigr)
      =\frac1{\sqrt N}\sum_a\ket a_A\,U_a\ket\psi_S .
  \end{aligned}
\end{equation}
The matrix elements of the reduced state of the register are the orbit
overlaps,
$\bra a\rho_A\ket b=\langle\psi|U_b^\dagger U_a|\psi\rangle/N$, so
$\ket{\Omega_{\mathcal U}(\psi)}$ is maximally entangled across the
register--system cut $A{:}S$ exactly when $\{U_a\ket\psi\}_a$ is an
orthonormal basis.  Equivalently, $\sum_aU_a\ket\psi\bra a$ is then
unitary and $\ket{\Omega_{\mathcal U}(\psi)}$ is its Choi state, a
resource for gate teleportation \cite{Kraus2009}.

What separates the two columns of
\Cref{tab:parallel-characterizations} is where this control lives.  For
an independent encoding, $C_{\mathcal U}$ factorizes, after grouping
tensor factors by party, as
\[
  C_{\mathcal U}
  =\bigotimes_{j=1}^n
    \left(
      \sum_{a_j=0}^1
      \ketbra{a_j}{a_j}_{A_j}\otimes U_j(a_j)
    \right),
\]
one controlled gate between each party's qubit and its own ancilla: this is
exactly the LME protocol of Ref.~\cite{Kraus2009}.  A commuting-reflection
encoding still produces maximal register--system entanglement, but its
control need not factorize in this way.

\section{Separating the encoder classes:
\texorpdfstring{$W_3$}{W3} and \texorpdfstring{$M_4$}{M4}}
\label{sec:structured-examples}
This section gives explicit encoders for $\ket{W_3}$ and $\ket{M_4}$.
They turn both inclusions in
\eqref{eq:strict-four-qubit-hierarchy} into strict ones: the $W$ state is
not LME yet admits a commuting-reflection encoding, and $M_4$ admits a
general local encoding but no commuting-reflection one.

\subsection{The \texorpdfstring{$W$}{W} state:
coordination beyond the LME class}
\label{sec:w-state}

The three-qubit $W$ state
\begin{equation}\label{eq:w-state}
  \ket{W_3}=\frac1{\sqrt3}
  \left(\ket{001}+\ket{010}+\ket{100}\right)
\end{equation}
is a decisive test of whether graph-basis freedom changes the state class.
Together with the GHZ state, it represents one of the two inequivalent forms
of genuine three-qubit entanglement under stochastic local operations
\cite{Duer2000}, and it is not LME \cite{Kraus2009}.  It therefore cannot be
encoded by independent binary choices.  Several $W$-type bases are nevertheless known
\cite{Miyake2005,Tanaka2007,Pimpel2023}.  Our claim is not a first
$W$-type basis: it is that $W_3$ already admits one inside the restricted
commuting-reflection class.

\begin{proposition}[A commuting-reflection $W$ basis]
\label{prop:w-encoding}
Let
\[
  q=i(bY+aZ),\qquad
  a^2=\frac{\sqrt3+1}{2\sqrt3},\qquad
  b^2=\frac{\sqrt3-1}{2\sqrt3},
\]
where $a,b>0$,
and set $V=q^{\otimes3}$.  Let $G_\triangle$ be the triangle graph, with
stabilizer generators
\[
  K_1=XZZ,\qquad K_2=ZXZ,\qquad K_3=ZZX.
\]
Then
\begin{equation}\label{eq:w-basis}
  \left\{V^\dagger K_1^{s_1}K_2^{s_2}K_3^{s_3}V\ket{W_3}:
  s\in\{0,1\}^3\right\}
\end{equation}
is an orthonormal basis of $W$-type states generated by commuting product
reflections.
\end{proposition}

The encoders are the conjugated stabilizer generators
$E_i=V^\dagger K_iV=-X_i\widetilde Z_j\widetilde Z_k$, built from
$\widetilde Z=q^\dagger Zq=cY+dZ$ with $c=\sqrt{2/3}$ and
$d=1/\sqrt3$.  They commute and square to the identity because they are
unitary conjugates of stabilizer generators, so only orthogonality of their
orbit remains to be checked.  Permutation symmetry of $\ket{W_3}$ reduces
that to one expectation value per Hamming weight of $s$.  Two of the three
vanish because $E_1$ and $E_1E_2E_3$ either change the excitation number or
have imaginary matrix elements in a real state; the remaining one vanishes
precisely when $c^2=2d^2$, which is what fixes $a$ and $b$.  The full
computation is carried out in \cref{app:w3-m4-proofs}.

\begin{corollary}\label{cor:strict-inclusion}
Independent encoders are strictly weaker than commuting-reflection encoders:
\[
  \EI\subsetneq\ECI.
\]
\end{corollary}

\begin{proof}
The inclusion was established after \cref{thm:kraus}.  It is strict because
$\ket{W_3}\notin\EI$ by \cref{thm:kraus}, whereas
$\ket{W_3}\in\ECI$ by \cref{prop:w-encoding}.
\end{proof}

The separation persists when the system is enlarged, which is what allows us
to state the whole hierarchy \eqref{eq:strict-four-qubit-hierarchy} at four
qubits.

\begin{remark}[A four-qubit witness for the first separation]
\label{rem:four-qubit-w-witness}
The state $\ket{W_3}\otimes\ket0$ lies in
$\ECI^{(4)}\setminus\EI^{(4)}$.
\end{remark}

Indeed, appending the two encoders $\openone$ and $X$ on the fourth
qubit to the eight commuting reflections of \cref{prop:w-encoding} gives
sixteen commuting product reflections whose orbit is an orthonormal basis,
so the state lies in $\ECI^{(4)}$.  It is not in $\EI^{(4)}$: by
\cref{thm:kraus} that would make it flat in some product basis, and a
product state is flat in a product basis only if each of its factors is,
which would force $\ket{W_3}$ to be LME.

\subsection{The \texorpdfstring{$M_4$}{M4} basis: exact encoding and
strict separation}
\label{sec:m4-separator}

We now show that commuting reflections, in turn, are strictly weaker than
general local encoders.  The witness is the four-qubit
Higuchi--Sudbery state
\begin{equation}\label{eq:m4-state}
  \begin{aligned}
  \ket{M_4}=\frac1{\sqrt6}\Bigl[&
  \ket{0011}+\ket{1100}\\[-2pt]
  &+\omega(\ket{1010}+\ket{0101})\\[-2pt]
  &+\omega^2(\ket{1001}+\ket{0110})\Bigr],
  \qquad \omega=e^{2\pi i/3}.
  \end{aligned}
\end{equation}
For a real four-vector
$\bm q=(q_0,q_1,q_2,q_3)$ of unit length, define the single-qubit unitary
\begin{equation}\label{eq:quaternion-gate}
  Q(\bm q)=q_0\openone-i(q_1X+q_2Y+q_3Z).
\end{equation}
Set $s=1/\sqrt2$ and $h=1/2$.  For row $r$ of
\cref{tab:m4-encoders}, let
\begin{equation}\label{eq:m4-encoders}
  U_r=Q(\bm q_{r1})\otimes Q(\bm q_{r2})\otimes Q(\bm q_{r3})
      \otimes\openone,
\end{equation}
so that the fourth qubit is untouched.\footnote{This is no loss of
generality here.  The state $\ket{M_4}$ is a four-qubit $SU(2)$
singlet and is therefore invariant under $v^{\otimes4}$ for every
$v\in SU(2)$ \cite{Higuchi2000,Kempe2001,Bourennane2004}.  Replacing any
encoder $U$ by $U\,v^{\otimes4}$, with $v$ the inverse of its fourth
factor, leaves the encoded state unchanged and trivializes that factor.}

\begin{proposition}[The $M_4$ basis]\label{prop:m4-encoding}
The sixteen states
$\{U_r\ket{M_4}:r=1,\ldots,16\}$ form an orthonormal basis.  Every
one-qubit gate in \cref{tab:m4-encoders} is a Clifford unitary.  The
sixteen product unitaries contain a pair that does not commute even up to an
overall phase.
\end{proposition}

\begin{proof}
Every vector listed in \cref{tab:m4-encoders} has unit norm, so each
factor \eqref{eq:quaternion-gate} is unitary.  The three claims are then
finite exact computations in the field generated by $\sqrt2$, $\sqrt3$,
and $i$, which we carried out symbolically; the scripts are in the
companion repository.  First,
\begin{equation}\label{eq:m4-gram}
  \bra{M_4}U_r^\dagger U_t\ket{M_4}=\delta_{rt},
\end{equation}
and sixteen orthonormal states in a sixteen-dimensional Hilbert space form
a complete basis.  Second, conjugating $X,Y,Z$ by each listed local factor
returns a signed Pauli matrix, which is the single-qubit Clifford
condition.  Third, taking rows 2 and 3, the product $U_2U_3$ is not
proportional to $U_3U_2$.
\end{proof}

\refstepcounter{table}\label{tab:m4-encoders}
\begin{center}
\begin{minipage}{\columnwidth}
\begingroup
\centering
\scriptsize\upshape
\setlength{\tabcolsep}{1.7pt}
\renewcommand{\arraystretch}{1.03}
\begin{tabular}{@{}c|c|c|c@{}}
$r$&$\bm q_{r1}$&$\bm q_{r2}$&$\bm q_{r3}$\\
\hline
1 &(1,0,0,0)&(1,0,0,0)&(1,0,0,0)\\
2 &(0,1,0,0)&(0,0,1,0)&(0,1,0,0)\\
3 &$(h,-h,-h,-h)$&$(h,h,h,h)$&$(s,0,0,s)$\\
4 &$(s,0,0,-s)$&$(h,-h,h,-h)$&$(h,h,-h,h)$\\
5 &$(0,0,s,-s)$&(1,0,0,0)&$(0,0,s,s)$\\
6 &$(s,-s,0,0)$&(0,0,0,1)&$(s,s,0,0)$\\
7 &$(0,s,s,0)$&$(h,h,-h,-h)$&$(h,-h,h,h)$\\
8 &(0,0,1,0)&(0,0,0,1)&(0,0,1,0)\\
9 &$(0,0,s,s)$&(0,0,1,0)&$(0,0,s,-s)$\\
10&$(s,0,0,s)$&$(h,-h,-h,h)$&$(h,h,h,-h)$\\
11&(0,0,0,1)&(0,1,0,0)&(0,0,0,1)\\
12&$(h,h,h,-h)$&$(h,-h,-h,h)$&$(0,s,-s,0)$\\
13&$(0,s,-s,0)$&$(h,h,h,h)$&$(h,-h,-h,-h)$\\
14&$(h,-h,h,h)$&$(h,h,-h,-h)$&$(s,0,0,-s)$\\
15&$(s,s,0,0)$&(0,1,0,0)&$(s,-s,0,0)$\\
16&$(h,h,-h,h)$&$(h,-h,h,-h)$&$(0,s,s,0)$
\end{tabular}
\par\medskip
\raggedright\small
\textsc{Table \thetable.}\quad
Four-vector parameters for the sixteen coordinated local-Clifford encoders of
$\ket{M_4}$, with $s=1/\sqrt2$ and $h=1/2$.
\endgroup
\end{minipage}
\end{center}

\begin{proposition}[The $M_4$ separator]\label{prop:m4-separator}
The Higuchi--Sudbery state admits no commuting-reflection encoding.
\end{proposition}

The reason can be seen without solving for any encoder.  By
\cref{thm:commuting-characterization}, such an encoding would make $M_4$
flat in some graph-state basis after a suitable local change of
basis, and flatness forces the expectation value of
every nonidentity graph stabilizer to vanish, because every such
stabilizer has equally many $+1$ and $-1$ eigenvectors in the graph basis.
The local change of basis turns the three Pauli directions of each qubit
into an unknown orthonormal frame of \emph{axes} in ordinary
three-dimensional space, and each letter of a stabilizer selects one axis
of the corresponding frame.  The state $M_4$ is a four-qubit $SU(2)$
singlet, so it is invariant under identical rotations $v^{\otimes4}$
\cite{Higuchi2000,Kempe2001,Bourennane2004}.  Consequently, its
correlations see these axes only through
rotation-invariant combinations; for instance, for any two qubits and unit
vectors $\mathbf a,\mathbf b$,
$\langle\sigma_{\mathbf a}\sigma_{\mathbf b}\rangle
=-\tfrac13\,\mathbf a\cdot\mathbf b$, where
$\sigma_{\mathbf a}=a_xX+a_yY+a_zZ$.  For $M_4$, a vanishing
two-qubit expectation means that two ordinary three-dimensional axes are
perpendicular, while a vanishing three-qubit expectation means that three
such axes lie in one plane.  Local-Clifford changes and relabelling reduce
all $64$ four-vertex graphs to six inequivalent types.  The reduction
is free because single-qubit Cliffords are themselves local unitaries,
absorbed into the arbitrary basis change, and they act on graph states by
local complementation.  Together with
the corresponding four-qubit correlation rule these elementary constraints
exclude all six possible graphs(\cref{fig:m4-graph-classes}).  For example, the star graph would require three axes to be
mutually perpendicular and a fourth axis to be perpendicular to all three.
In the hardest case, the four-vertex path, the perpendicularity and
coplanarity conditions force two four-qubit stabilizers to obey
$\langle YXXY\rangle^{2}+\langle YXYZ\rangle^{2}=1/9$, so their
expectations cannot both vanish.
The full case analysis is given in \cref{app:w3-m4-proofs}.

\input{figures/m4_graph_classes}

Combining the two propositions gives the exact separation
\begin{equation}\label{eq:m4-separation}
  \ket{M_4}\in\EG^{(4)}\setminus\ECI^{(4)},
\end{equation}
and hence $\ECI^{(4)}\subsetneq\EG^{(4)}$.  We had originally studied
$\ket{M_4}$ as a candidate for a nonencodable state; it turned out to be
encodable, but only with the full freedom of general local encoders.

The example is independently notable for its entanglement.  It globally
maximizes both the geometric measure
\cite{Derksen2017,DenkerImaiGuhne2025} and the average two-versus-two linear
entropy \cite{Gour2010}.  Higuchi and Sudbery originally proposed it using
the corresponding average von Neumann entropy; it is a proven local maximum
and a global maximum within the 1-uniform family for that functional, while
unrestricted global optimality remains conjectural
\cite{Higuchi2000,Brierley2007,Gour2010}.  Its largest squared overlap with a
fully product state is $2/9$, the minimum possible for four qubits.  Thus
all sixteen locally equivalent basis vectors share the global geometric- and
linear-entropy extremality.  This extremality has an immediate operational
consequence: a global measurement in the $M_4$ basis identifies a
uniformly encoded four-bit message perfectly, whereas any LOCC measurement,
and indeed any fully separable one, succeeds with probability at most
$2/9$ (see \cref{app:w3-m4-proofs}).  The basis therefore hides classical
data from local observers in the spirit of
Refs.~\cite{DiVincenzo2002,Hayashi2006}.  Recent bipartite work similarly
treats shared entanglement as a resource for locally distinguishing maximally
entangled bases \cite{BandyopadhyayRusso2024}; the $M_4$ basis merits
further study from this perspective.

The $W$ and $M_4$ examples show that widening the allowed encoder
class really does enlarge the accessible set of states.  We now remove all
restrictions on the product encoders and ask whether general local encoding
is universal.

\section{A four-qubit obstruction to unrestricted local encoding}
\label{sec:nonencodability}

Our goal in this section is to prove that some four-qubit states admit no
local encoding at all, however the sixteen product unitaries are chosen.

The strategy is a comparison of dimensions.  We use a regular subset of the
four-qubit 1-uniform states---states whose one-qubit marginals are all
maximally mixed---which forms a smooth 18-dimensional family, and prove that
the encodable states within it have at most 15 independent parameters.  They
therefore cannot fill the family.

Let us first explain why such a bound cannot be obtained by simply counting
equations and unknowns.  An encoding involves not only the state but also
sixteen auxiliary product unitaries; we call this list, with one
product unitary per basis element, an \emph{encoder table}.  The encoders contain many
parameters and many redundancies: a single state may admit a vast set of
different encoder tables.  Subtracting the number of orthogonality
constraints from the number of encoder parameters therefore says nothing
about how many \emph{states} are encodable---which is why the count in
Ref.~\cite{Pimpel2023} could only be heuristic.  What is needed is a
constraint on the state alone: one that holds no matter which encoders are
used, and no matter how the encoders readjust as the state changes.

To see the mechanism, compare two infinitesimal 
changes of a normalized state $\ket{\psi}$, 
\begin{equation}
\ket{\psi} \mapsto \ket{\psi} + \varepsilon \ket{v}
\mbox{ and }
\ket{\psi} \mapsto \ket{\psi} + \delta \ket{w}
\end{equation}
represented by ordinary first derivatives $\ket v$ 
and $\ket w$ after fixing the irrelevant global phase.  
Note that these infinitesimal changes do not need to be
generated by a unitary transformation. 

To each such pair we will assign a single number, the \emph{signed area}\footnote{Up to normalization conventions, this is the Fubini--Study
symplectic form on projective Hilbert space; for its role in the geometry of
entanglement and local-unitary or fixed-marginal manifolds, see
Refs.~\cite{Sawicki2011,SawickiOszmaniec2014,DulianSawicki2025}.} 
\begin{equation}
\omega_\psi(v,w)=2\operatorname{Im}\langle v|w\rangle,
\end{equation}
for the one-dimensional case, that is a 
single complex coefficient, it is the oriented
area spanned by the two changes in the complex plane, 
and for a state vector the contributions of all 
coefficients are added.

Two facts about this quantity drive the proof.  First, when 
an orthonormal basis varies smoothly, the signed areas for any pair of the variations, contributed by its $N$ elements
always sum to zero.  This is the ONB-cancellation lemma,
\cref{lem:onb-cancellation}. It involves no locality at 
all and holds because a small change of an ONB has to be
unitary and is thus generated by a Hermitian matrix,
exactly as in ordinary unitary time evolution. Then, the
statement follows because the signed area takes only the 
imaginary parts into account.

Second, consider a 1-uniform state $\ket{\psi}$ which is locally encodable, and consider variations of this state
in two directions, which are also locally encodable. These
variations naturally require an adjustment of the unitaries 
uses for the encoding. This leads to variations of each of
the basis states in two directions, driven by two contributions: The variation of $\ket{\psi}$ and the adjustment of the unitaries. For a basis state 
$\ket{\phi_a}=U_a\ket\psi$, undoing $U_a$ leads to the 
same two state changes $v,w$ for every $a$, plus local-unitary readjustments. In this situation the local-adjustment lemma, \cref{lem:local-null} shows that 1-uniformity implies that the readjustments of the encoding make no contribution
to the signed area. Every
basis vector $\ket{\phi_a}$ therefore contributes the same signed area $\omega_\psi(v,w)$.
The equal-contributions identity, \cref{prop:equal-contributions}, combines
the two facts: $N$ equal contributions with zero sum must each vanish, so
the signed area of any two changes within such a family is zero.  

However, for any direction $\ket{v}$ the signed area does not vanish if we take as a second direction $\ket{w} = i \ket v$.
This allows to find for any state $\ket{\psi}$ in the manifold of one-uniform states pairs of variations within this
manifold where the signed area does not vanish. Consequently,
\cref{sec:dimension-count} turns the zero-area condition 
into a bound of at most 15 independent state directions by a dimension count. \Cref{fig:no-go-overview} summarizes the argument.

\input{figures/no_go_overview}

The result can now be stated precisely.  Here
$\mathcal Z_4^{\mathrm{reg}}$ denotes the four-qubit 1-uniform states at which
the marginal conditions are independent, made precise below, and $\EG$ is the
set of locally encodable states.  The encodable states, together with their
encoder tables, need not form one smooth family: the set may contain
exceptional points at which the number of independent changes is different.
The dimension of such a set means the largest dimension of any of its
smooth pieces; \cref{app:technical-dimension} explains why treating the
pieces separately is sufficient.

\begin{theorem}[Non-universality of local encodings]
\label{thm:nonencodable}
For four qubits,
\begin{align}\label{eq:main-dimension-bound}
  \dim_{\mathbb R}\mathcal Z_4^{\mathrm{reg}}&=18,\nonumber\\
  \dim_{\mathbb R}
  \bigl(\EG\cap\mathcal Z_4^{\mathrm{reg}}\bigr)&\le15.
\end{align}
In particular, the encodable states do not fill
$\mathcal Z_4^{\mathrm{reg}}$, so nonencodable four-qubit pure states exist.
\end{theorem}

\subsection{The 1-uniform test family}

For $N=2^n$, let
\[
  \mathcal M_n=\mathbb{CP}^{N-1}
\]
be the space of pure $n$-qubit rays, and let
$\rho_j=\Tr_{\{1,\ldots,n\}\setminus\{j\}}\ketbra{\psi}{\psi}$
denote the reduced state of qubit $j$.  Define
  \begin{equation}\label{eq:uniform-family}
  \mathcal Z_n=
  \left\{[\psi]\in\mathcal M_n:
  \rho_j=\frac{\openone}{2}\text{ for }j=1,\ldots,n\right\}.
\end{equation}
These are also called the \emph{1-uniform} states
\cite{Scott2004,Goyeneche2014,Goyeneche2015}.  Equivalently, the $3n$
local Bloch expectations vanish:
\begin{equation}\label{eq:bloch-map}
  \bm b([\psi])=
  \bigl(
  \langle X_1\rangle,\langle Y_1\rangle,\langle Z_1\rangle,
  \ldots,
  \langle X_n\rangle,\langle Y_n\rangle,\langle Z_n\rangle
  \bigr)=0.
\end{equation}
Two features make $\mathcal Z_n$ a promising test family.  It is large, as
computed below, and its vanishing local Bloch expectations will make
local basis readjustments invisible to the signed area introduced in
\cref{sec:imaginary-overlap}.  

We call a state \emph{regular} when, to first order, these $3n$ conditions
impose $3n$ independent requirements.  Around any such state, the
1-uniform states form a smooth family of dimension
\begin{equation}\label{eq:z-dimension}
  \dim_{\mathbb R}\mathcal Z_n^{\mathrm{reg}}
  =2(N-1)-3n.
\end{equation}
Regular states exist: \cref{app:explicit-state} verifies regularity within
an explicit eight-parameter family of 1-uniform four-qubit states, and in
particular at the state $\ket{\Psi_{\mathrm{ex}}}$ constructed there.
Hence $\mathcal Z_4^{\mathrm{reg}}$ is nonempty and has dimension
$30-12=18$.

By a phase-fixed derivative $v$ we mean an ordinary normalized state path
$\ket{\psi(s)}=\ket\psi+s\ket v+O(s^2)$ whose phase has been chosen so that
$\langle\psi|v\rangle=0$.  At a regular $n$-qubit state, let $E_\psi$ be the
real vector space of all such derivatives and let $V_\psi\subset E_\psi$
contain those that preserve 1-uniformity to first order.  Equivalently,
\[
  \begin{aligned}
    E_\psi
      &=\{v\in\mathbb C^N:\langle\psi|v\rangle=0\}
        =T_{[\psi]}\mathcal M_n,\\
    V_\psi
      &=T_{[\psi]}\mathcal Z_n^{\mathrm{reg}}.
  \end{aligned}
\]
Their real dimensions are $2(N-1)$ and $2(N-1)-3n$, respectively.
Let $P_\mu$, $\mu=1,\ldots,3n$, denote the one-qubit Pauli operators
$X_1,Y_1,Z_1,\ldots,X_n,Y_n,Z_n$.  The first-order change of a Pauli
expectation value along $v$ is
\begin{equation}\label{eq:pauli-first-change}
  \delta_v\langle P_\mu\rangle_\psi
  =2\operatorname{Re}\langle P_\mu\psi|v\rangle,
\end{equation}
so $v$ belongs to $V_\psi$ exactly when all $3n$ of these expressions
vanish.  For the same $3n$ operators, define
\begin{equation}\label{eq:local-directions}
  k_\mu=-iP_\mu\ket\psi,
  \qquad
  K_\psi=\operatorname{span}_{\mathbb R}\{k_1,\ldots,k_{3n}\}.
\end{equation}
Local unitaries preserve 1-uniformity, so $K_\psi\subseteq V_\psi$, and
regularity makes the $3n$ local-rotation directions independent: if a real
combination $R=\sum_\mu r_\mu P_\mu$ satisfied $R\ket\psi=0$, then by
\eqref{eq:pauli-first-change} the same combination of the $3n$ first-order
marginal conditions would vanish on every $v$, contradicting regularity.
The directions in $K_\psi$ merely reorient the local
bases.  For four qubits, $\dim V_\psi=18$ and $\dim K_\psi=12$, leaving six
directions that change the local-unitary orbit and hence the entanglement
type---the multipartite analogue of changing Schmidt coefficients rather than
changing the local Schmidt bases.  The proof below identifies the concrete
six-dimensional complement $C_\psi$.

\subsection{Why encoded 1-uniform families have zero signed area}
\label{sec:imaginary-overlap}

To bound the dimension of an encodable family, it is not enough to examine
one change at a time.  We must ask which pairs of changes can coexist while
the state and all its encoders continue to vary smoothly.  We first define
the signed area precisely and then prove the two facts announced above.

A pure state is a ray, so a change proportional to $i\ket\psi$ merely changes
its phase.  For a smooth two-parameter family, choose the phase so that, at
the point under consideration, both derivatives are phase-fixed:
\begin{equation}\label{eq:overview-variations}
  \begin{aligned}
    \ket{\psi(s,t)}
      &=\ket\psi+s\ket v+t\ket w+O(\|(s,t)\|^2),\\
    v&=\left.\partial_s\psi\right|_0,
    \qquad
    w=\left.\partial_t\psi\right|_0,\\
    \langle\psi|v\rangle&=\langle\psi|w\rangle=0.
  \end{aligned}
\end{equation}
For two changes written in this convention, define

\begin{equation}\label{eq:omega}
  \omega_\psi(v,w)=2\operatorname{Im}\langle v|w\rangle.
\end{equation}

This quantity is bilinear over the real numbers and antisymmetric:
$\omega_\psi(w,v)=-\omega_\psi(v,w)$.  For a single complex coordinate it
is, up to the factor of two, the ordinary signed area in the corresponding
real plane; with several coordinates, these signed contributions are added.
This motivates the name.  The value is unchanged if a different
representative of the ray is used.  In the full tangent space
$E_\psi$, every nonzero direction $v$ has a positive-area partner $iv$:
\begin{equation}\label{eq:omega-nonzero-partner}
  \omega_\psi(v,iv)=2\langle v|v\rangle>0.
\end{equation}
This statement involves only the full space $E_\psi$: if $v$ preserves
1-uniformity, $iv$ does not need to do so.  \Cref{sec:dimension-count}
characterizes exactly when both directions do.

An infinitesimal local-basis adjustment moves $\ket\psi$ along
\begin{equation}\label{eq:local-generator}
  -iR\ket\psi,\qquad
  R=\sum_{\mu=1}^{3n}r_\mu P_\mu,
\end{equation}
where $R$ is a sum of single-qubit Hermitian generators.  If $w$
preserves 1-uniformity, then by \eqref{eq:pauli-first-change}
\begin{equation}\label{eq:mixed-area-zero}
  \omega_\psi(-iR\psi,w)
  =2\operatorname{Re}\langle R\psi|w\rangle
  =\delta_w\langle R\rangle_\psi=0.
\end{equation}
The area between two local adjustments also vanishes:
\begin{equation}\label{eq:local-area-zero}
  \omega_\psi(-iR\psi,-iS\psi)
  =\frac1i\langle[R,S]\rangle_\psi=0.
\end{equation}
Indeed, $[R,S]/i$ is again a real linear combination of one-qubit Pauli
operators, all of whose expectations vanish in a 1-uniform state.

\begin{lemma}[Local basis adjustments make no contribution]
\label{lem:local-null}
Let $\ket\psi$ be 1-uniform, and let $v,w$ preserve 1-uniformity to
first order.  For arbitrary local generators $R,S$,
\begin{equation}\label{eq:local-no-contribution}
  \omega_\psi(v-iR\psi,w-iS\psi)=\omega_\psi(v,w).
\end{equation}
\end{lemma}

\begin{proof}
Expand the left-hand side.  The two mixed terms vanish by
\eqref{eq:mixed-area-zero}, and the local--local term vanishes by
\eqref{eq:local-area-zero}.
\end{proof}

We next use a property of complete orthonormal bases that does not involve
locality.  The only assumption is that each of the two changes preserves all
orthonormality relations to first order.

\begin{lemma}[ONB cancellation]\label{lem:onb-cancellation}
Let $\{\ket{\phi_a}\}_{a=1}^N$ be an ordered orthonormal basis.  Consider two
first-order changes that each preserve orthonormality,
$\{\dot\phi_a\}_a$ and $\{\dot\phi'_a\}_a$.  Then
\begin{equation}\label{eq:onb-cancellation}
  \sum_{a=1}^N
  \omega_{\phi_a}(\dot\phi_a,\dot\phi'_a)=0.
\end{equation}
\end{lemma}

\begin{proof}

At the point in question, choose the phase of each basis vector so that its
derivatives in both parameter directions satisfy the convention of
\eqref{eq:overview-variations}.  The two phase derivatives can be chosen
independently.
Collect the basis vectors as the columns of a unitary matrix
$\Phi=[\,\phi_1\ \cdots\ \phi_N\,]$.  Preserving orthonormality to first
order means exactly that the two variations have the form
$\dot\Phi=\Phi A$ and $\dot\Phi'=\Phi B$ with $A$ and $B$ anti-Hermitian.
Indeed, differentiating $\Phi^\dagger\Phi=I$ gives
$\Phi^\dagger\dot\Phi+(\Phi^\dagger\dot\Phi)^\dagger=0$, so
$A=\Phi^\dagger\dot\Phi$ satisfies $A^\dagger=-A$ and
$\dot\Phi=\Phi A$; the same argument gives $\dot\Phi'=\Phi B$.
Write $A=-iH$ and $B=-iG$, with $H,G$ Hermitian, so that
$A^\dagger B=HG$.  Summing the areas of all columns gives a trace,
\begin{align}
  \sum_a\omega_{\phi_a}(\dot\phi_a,\dot\phi'_a)
  &=2\operatorname{Im}\Tr(A^\dagger B)\nonumber\\
  &=2\operatorname{Im}\Tr(HG)=0,
\end{align}
because the trace of a product of two Hermitian matrices is real.

\end{proof}

We now combine the two facts.  Consider a smooth two-parameter family in
which both the 1-uniform fiducial state and the encoder tables vary.  The
following proposition shows that the signed area of the two state changes
must then be zero.

\begin{proposition}[The equal-contributions identity]
\label{prop:equal-contributions}
Suppose $\ket{\psi(s,t)}$ is a smooth two-parameter family of 1-uniform
states, and suppose smooth product-unitary families $U_a(s,t)$ make
\[
  \ket{\phi_a(s,t)}=U_a(s,t)\ket{\psi(s,t)},
  \qquad a=1,\ldots,N,
\]
an orthonormal basis for every $s,t$.  At $s=t=0$, choose the phase of
$\ket\psi$ and write
\[
  v=\partial_s\psi,\qquad w=\partial_t\psi,
  \qquad
  \langle\psi|v\rangle=\langle\psi|w\rangle=0.
\]
Then
\begin{equation}\label{eq:central-zero}
  \omega_\psi(v,w)=0.
\end{equation}
\end{proposition}

\begin{proof}
Differentiate $\ket{\phi_a}=U_a\ket\psi$ in both directions and multiply
by $U_a^\dagger$ at the point under consideration; this does not change
inner products.  With
$R_a=iU_a^\dagger(\partial_sU_a)$, and $S_a$ defined in the same way,
\begin{align}\label{eq:element-variations}
  U_a^\dagger\partial_s\phi_a&=v-iR_a\psi,\nonumber\\
  U_a^\dagger\partial_t\phi_a&=w-iS_a\psi,
\end{align}
where $R_a,S_a$ are sums of single-qubit Hermitian generators because
$U_a(s,t)$ is always a product unitary; the generators may be entirely
different for different values of $a$.  A multiple of the identity in
$R_a$ or $S_a$ would only change the phase of $\ket{\phi_a}$, so we drop
it.  This choice automatically enforces the phase convention
\eqref{eq:overview-variations} for every basis vector:
$\langle\phi_a|\partial_s\phi_a\rangle
=\langle\psi|v\rangle-i\langle R_a\rangle_\psi=0$, because $v$ is
phase-fixed and every one-qubit Pauli expectation vanishes in a 1-uniform
state.

Because $v,w$ preserve 1-uniformity, \cref{lem:local-null} says that the
only contribution to the area is from the change of the state, not from the
local basis adjustments:
\begin{align}\label{eq:equal-contributions}
  \omega_{\phi_a}(\partial_s\phi_a,\partial_t\phi_a)
  &=\omega_\psi(v-iR_a\psi,w-iS_a\psi)\nonumber\\
  &=\omega_\psi(v,w).
\end{align}
Thus all $N$ basis elements contribute the same area.  Since they remain
an ONB, \cref{lem:onb-cancellation} also says that these areas sum to zero:
\begin{align}\label{eq:key-cancellation}
  0
  &=\sum_{a=1}^N
    \omega_{\phi_a}(\partial_s\phi_a,\partial_t\phi_a)\nonumber\\
  &=N\,\omega_\psi(v,w),
\end{align}
which proves the claim.
\end{proof}

Neither ingredient alone would constrain the fiducial state; the result
comes from combining them.

\subsection{From zero signed area to the dimension bound}
\label{sec:dimension-count}

We now identify the six directions for which the positive-area pair
$v,iv$ remains inside the 1-uniform family.  For the rotated direction
$iv$, formula \eqref{eq:pauli-first-change} gives
\[
  \delta_{iv}\langle P_\mu\rangle_\psi
  =-2\operatorname{Im}\langle P_\mu\psi|v\rangle.
\]
Consequently, both $v$ and $iv$ preserve 1-uniformity exactly when both the
real and imaginary parts vanish:
\begin{equation}\label{eq:both-uniform}
  v,iv\in V_\psi
  \quad\Longleftrightarrow\quad
  \langle P_\mu\psi|v\rangle=0
  \quad\text{for every }\mu.
\end{equation}
Collect these directions in the space
\begin{equation}\label{eq:paired-space}
  C_\psi=V_\psi\cap iV_\psi
  =\{v\in V_\psi:iv\in V_\psi\},
\end{equation}
which is closed under multiplication by $i$: if $v,iv\in V_\psi$, then also
$iv,i(iv)=-v\in V_\psi$.  Equation~\eqref{eq:both-uniform} gives the
equivalent characterization
\[
  C_\psi
  =\{v:\langle\psi|v\rangle=0,
          \ \langle P_\mu\psi|v\rangle=0\text{ for every }\mu\}.
\]

Let us determine its dimension using only dimensions of ordinary real vector
spaces.  Both $V_\psi$ and $iV_\psi$ have dimension $18$ and lie inside the
$30$-dimensional space $E_\psi$.  Since
\[
  \dim(V_\psi+iV_\psi)
  =18+18-\dim C_\psi\le30,
\]
we must have $\dim C_\psi\ge6$.

There cannot be more than six such directions.  Recall that the twelve local
rotations form the subspace $K_\psi\subset V_\psi$.  Suppose that
$k\in K_\psi\cap C_\psi$.  Then $k=-iR\ket\psi$ for some local generator
$R$ and $ik\in V_\psi$, so \eqref{eq:mixed-area-zero} gives
$\omega_\psi(k,ik)=0$.  On the other hand,
\eqref{eq:omega-nonzero-partner} gives
$\omega_\psi(k,ik)=2\langle k|k\rangle$, so $k=0$.  Thus
$K_\psi\cap C_\psi=\{0\}$.  These two subspaces both lie in the
18-dimensional $V_\psi$, and therefore
\[
  12+\dim C_\psi
  =\dim(K_\psi+C_\psi)\le18.
\]
Together with the lower bound, this proves
\begin{equation}\label{eq:paired-space-dimension}
  \dim_{\mathbb R}C_\psi=6,
  \qquad
  V_\psi=K_\psi\mathbin\oplus C_\psi.
\end{equation}

Now consider one smooth family in which both the state and its encoder table
vary.  Let $L\subseteq V_\psi$ be the space of first-order state changes
produced within this family, and set
\[
  d=\dim L,
  \qquad
  W=L\cap C_\psi.
\]
The equal-contributions identity says
\begin{equation}\label{eq:area-zero-on-L}
  \omega_\psi(v,w)=0
  \qquad\text{for every }v,w\in L.
\end{equation}

We need two elementary bounds on the dimension of $W$.  First, both $L$ and
$C_\psi$ lie in the 18-dimensional space $V_\psi$, so
\begin{align}
  \dim(L+C_\psi)
  &=\dim L+\dim C_\psi-\dim(L\cap C_\psi)\nonumber\\
  &=d+6-\dim W\le18.
\end{align}
Rearranging gives
\begin{equation}\label{eq:W-lower-bound}
  \dim W\ge d-12.
\end{equation}
In words, $C_\psi$ has codimension twelve inside $V_\psi$, so requiring a
direction in $L$ to lie in $C_\psi$ can remove at most twelve independent
directions from $L$.

Second, $W$ cannot contain a nonzero direction together with its multiple by
$i$.  Indeed, if $x\in W\cap iW$ were nonzero, then $x=iy$ for some nonzero
$y\in W$.  Both $y$ and $iy=x$ would belong to $L$, so
\eqref{eq:area-zero-on-L} would give $\omega_\psi(y,iy)=0$, contradicting
$\omega_\psi(y,iy)=2\langle y|y\rangle>0$.  Hence
$W\cap iW=\{0\}$.  Multiplication by $i$ does not change dimension, and both
$W$ and $iW$ lie in the six-dimensional $C_\psi$.  Therefore
\begin{equation}\label{eq:W-upper-bound}
  2\dim W
  =\dim(W+iW)\le6,
  \qquad
  \dim W\le3.
\end{equation}

Combining \eqref{eq:W-lower-bound} and \eqref{eq:W-upper-bound} gives
\[
  d-12\le\dim W\le3,
\]
and hence
\begin{equation}\label{eq:smooth-family-bound}
  d\le15.
\end{equation}

This bounds the state directions of one smoothly varying state and encoder
table.  As explained in \cref{app:technical-dimension}, the joint solution
set of states and encoders divides into finitely many smooth pieces, and
the bound applies to every piece, so
$\dim_{\mathbb R}(\EG\cap\mathcal Z_4^{\mathrm{reg}})\le15$.  Since
$\mathcal Z_4^{\mathrm{reg}}$ has dimension 18, the encodable states cannot
fill it, and nonencodable states exist.  This proves
\cref{thm:nonencodable}.

\begin{remark}[General $n$ and the four-qubit threshold]
\label{rem:threshold}
Nothing in the argument is special to four qubits.  At a regular point,
$\dim E_\psi=2(N-1)$, $\dim V_\psi=2(N-1)-3n$, and $\dim K_\psi=3n$; the
inclusion $K_\psi\subseteq V_\psi$ already requires $N-1\ge3n$ there.  The
identical argument gives $\dim C_\psi=2(N-1-3n)$, the two bounds on
$W=L\cap C_\psi$ become $\dim W\ge d-3n$ and $\dim W\le N-1-3n$, and every
smooth piece obeys $d\le N-1$.  Hence
\begin{equation}\label{eq:global-bound}
  \dim_{\mathbb R}
  \left(\EG\cap\mathcal Z_n^{\mathrm{reg}}\right)\le N-1,
\end{equation}
whereas the family itself, whenever nonempty, has dimension $2(N-1)-3n$.
The bound is an obstruction precisely when $N-1>3n$.  For qubits this
first occurs at $n=4$, where $15>12$.
\end{remark}

\subsection{An explicit nonencodable state}

The dimension argument proves that nonencodable regular 1-uniform states
exist, but by itself it does not identify one.  We now give an exact example.

\begin{proposition}\label{cor:explicit-nonencodable}
Let $\alpha=2^{1/7}$ and $\mathcal N=\sqrt{6+2\mathrm e^2}$.  The state
\begin{align}\label{eq:explicit-nonencodable}
  \ket{\Psi_{\mathrm{ex}}}=\frac1{\mathcal N}\bigl[&
  \ket{0000}+\mathrm e^{i\alpha}\ket{1111}\nonumber\\[-2pt]
  &+\mathrm e^{i\alpha^2}\ket{0011}
   +\mathrm e^{i\alpha^3}\ket{1100}\nonumber\\[-2pt]
  &+\mathrm e^{i\alpha^4}\ket{0101}
   +\mathrm e^{i\alpha^5}\ket{1010}\nonumber\\[-2pt]
  &+\mathrm e^{1+i\alpha^6}\ket{0110}
   +\mathrm e^{1+i\alpha^7}\ket{1001}\bigr]
\end{align}
is 1-uniform and regular, and it admits no local encoding.
\end{proposition}

The 1-uniformity is visible directly.  All occupied strings have even
parity, so the off-diagonal entries of every one-qubit marginal vanish, and
the two amplitudes in each complementary pair have the same magnitude, so
the two diagonal entries are equal.

For the remaining claims, allow the seven phases in
\eqref{eq:explicit-nonencodable} to vary independently and replace the common
magnitude $\mathrm e$ of the last complementary pair by an arbitrary positive
number $r$.  This gives an eight-parameter family of 1-uniform states.  Two
particularly simple changes within it are: changing $r$, and changing the
phase of the $\ket{1111}$ coefficient.  \Cref{app:explicit-state} computes
their signed area directly and finds
\[
  -\frac{r}{(r^2+3)^2}\ne0.
\]
If a smooth encodable piece occupied all eight parameters, both changes would
belong to its tangent space, contradicting the equal-contributions identity.
Thus the encodable part has dimension at most seven.

The appendix then uses the polynomial description of states and encoders to
show that this lower-dimensional set obeys an additional nonzero polynomial
equation among the eight parameters.  The amplitudes in
\eqref{eq:explicit-nonencodable} are chosen precisely so that no such
equation can hold: by the Lindemann--Weierstrass theorem
\cite{Popescu2024}, the eight numbers
$\mathrm e,\mathrm e^{i\alpha},\ldots,\mathrm e^{i\alpha^7}$ satisfy no
nonzero polynomial relation with algebraic coefficients.  Hence the
displayed state cannot be encodable.  The exact argument, including
regularity, is given in \cref{app:explicit-state}.

\section{Discussion}\label{sec:discussion}

We have studied which types of multipartite entanglement can appear in a
complete projective measurement, and we have organized the question by the encoders used to generate the basis.  The resulting
picture has two main anchor points.  At the restricted end, the
commuting-reflection class is characterized exactly: the encodable states are
precisely those that are flat in a graph-state basis, a direct
generalization of the Kruszy\'nska--Kraus characterization of independent
encoders in which the product basis is replaced by an entangled stabilizer
eigenbasis [\cref{thm:commuting-characterization,cor:strict-inclusion}].
At the other extreme end, we showed that no encoder class is universal: for four
qubits, some 1-uniform states cannot generate any orthonormal basis under
product unitaries, and $\ket{\Psi_{\mathrm{ex}}}$ is an explicit example
[\cref{thm:nonencodable,cor:explicit-nonencodable}], answering the
question left open in 2007 \cite{Tanaka2007} and proving the later
four-qubit conjecture \cite{Pimpel2023}. The local-encodability problem for bipartite qudits is also being
investigated independently
\cite{EisfeldWyderkaInPreparation}.

We believe the most interesting message of this work lies between these
two anchor points: the perspective of studying the entanglement of
measurements \emph{as a function of the encoding class}.  Each class
$\mathcal C$ of encoder families defines a set $E_{\mathcal C}$ of
entanglement types---a local-unitary-invariant set of states---and the
independent and commuting-reflection classes show that this assignment is
informative: enlarging the class from independent choices to commuting
reflections enlarges the state class in a controlled, exactly
characterizable way, while
\cref{eq:m4-separation} shows that still more states become accessible once
more general encoders are allowed.  This immediately
raises what we consider the central open question.  \emph{Which other
encoding classes define natural state classes, and what do these sets look
like?}  Concrete intermediate classes could be: coordinated
tables of Pauli operators that need not commute, as in the $W$-basis of
Miyake and Briegel \cite{Miyake2005}; commuting families whose local factors
are not Hermitian, the smallest step beyond our theorem; and tables of local
Clifford unitaries, as in our $M_4$ construction.  For none of these is a
characterization known, although $M_4$ now proves that the last class is
strictly larger than commuting reflections.  More broadly, the
basic structure of the sets $E_{\mathcal C}$---their inclusions and
boundaries, useful invariants, and practical ways to decide
membership, remains largely unexplored beyond these two characterized
classes and the general results proved here.

The present hierarchy also connects naturally to the complementary
operational classification of Ref.~\cite{Pauwels2025}.  Here measurements
are organized by the encoding class needed to \emph{generate} their bases;
there they are organized by the entanglement needed to \emph{localize} and
implement them across distant laboratories.  The two viewpoints are
complementary in a literal sense: local encoding writes a classical label
into a shared state, whereas localization is the reverse task of recovering
that label locally with the aid of shared entanglement.  For bases generated
as orbits of a stabilizer group, this localizability analysis was recently
shown to simplify considerably \cite{Pauwels2026Pauli}.  Our
characterization theorem shows that, up to local unitaries, every basis
produced by a commuting-reflection encoding has precisely this form.  This
is therefore a broad class for which both the generation and implementation
structures can be analyzed systematically.  Relating the two hierarchies
beyond this class is a particularly promising direction.

The results also show that local encodability is not governed by any
ordinary notion of the \emph{amount} of entanglement.  The $W$ state is
excluded from the LME class yet admits a highly structured commuting
encoding.  The $M_4$ construction makes the complementary point: a state
that globally maximizes both the geometric measure and the average
two-versus-two linear entropy is encodable, but it already requires more
than commuting reflections.  Conversely, $\ket{\Psi_{\mathrm{ex}}}$ has
maximally mixed one-qubit marginals, a standard signature of strong
multipartite entanglement, yet cannot fill an orthonormal basis with its
local-unitary orbit.  What matters is the position of the entire orbit
relative to the orthogonality constraints of a basis, not a scalar property
of one representative.  Operationally, nonencodability is a limit on
dense-coding-type protocols: no local strategy can write $2^n$ classical
messages into such a state so that a single joint measurement reads them
out perfectly. Equivalently, local encodability is saturation of the first
Welch bound by $N$ states restricted to a single local-unitary orbit.  The
minimum excess above this bound is an orbit-constrained approximate-$1$-design
error of the type studied in Ref.~\cite{Castellano2026Designs}.

The four-qubit threshold is informative in the same spirit.  For three
qubits our dimension comparison gives no obstruction, and our numerical
searches consistently find a connected-graph-basis representation.  This
suggests a stronger form of the positive numerical picture of
Ref.~\cite{Pimpel2023}: \emph{Conjecture.---Every three-qubit pure state
admits a commuting-reflection local encoding.}

The proof technique behind \cref{thm:nonencodable} may be of independent
interest.  The signed-area identity constrains every smoothly varying
orthonormal basis, encoded or not: the signed areas contributed by the $N$
outcomes always cancel.  Whenever a family of bases is required to track a
family of states, this converts the orthogonality constraints into a
quantitative restriction on the states alone.  We expect the identity to
be useful in other basis-design problems, for instance for intermediate
encoder classes or in higher local dimensions.

Several concrete problems follow directly from our results.
First, viewed from the Bell-basis motivation, the $M_4$
construction is more than a class separator: it is a complete four-qubit
basis whose outcomes globally maximize two well-motivated multipartite
entanglement criteria, while its exact global-versus-local discrimination
gap makes it a concrete multipartite data-hiding basis.  It would be
interesting to determine whether it supports useful multipartite analogues
of the Bell basis's roles in teleportation, dense coding, entanglement
swapping, and quantum networks.  Second, the entire hierarchy should be developed for qudits,
where even the independent-encoder baseline is less understood.  Finally,
it would be valuable to bring the iso-entangled bases constructed here back
to the network scenarios where such measurements first appeared
\cite{Gisin2019,Tavakoli2021,Tavakoli2022}.

\begin{acknowledgments}
We thank Valerio Scarani and Nikolai Wyderka for interesting discussions. This research was funded by the Deutsche Forschungsgemeinschaft (DFG, German Research Foundation, project number 563437167), the Sino-German Center for Research Promotion (Project M-0294), and the German 
Federal Ministry of Research, Technology and Space 
(Project QuKuK, Grant No.\ 16KIS1618K and Project BeRyQC, Grant No.\ 13N17292) and the Swiss National Science Foundation
and NCCR-SwissMAP.

\emph{AI statement:}
We used Claude Fable and OpenAI's Codex tools
(GPT-o1 through GPT-5.6 Sol) to clarify statements, find references, explore and check proofs and candidate
constructions, check calculations, and assist with editing the final manuscript.  The authors independently verified every AI generated output. The main technical inputs of the AI were: it found and fixed a mistake in our initial proof of \cref{thm:commuting-characterization}, it found the reported $W$-state, $\eta$-state and $M$-state bases (remarkably, after extensive numerical exploration and Valerio
had failed to find one) and it greatly helped with \cref{thm:nonencodable} by proposing the area invariant for ONBs which eventually led to the proof.

\emph{Code availability.} The exact verification scripts and instructions are
available on
\href{https://github.com/jefpauwels/IsoEntangledBases}
{GitHub}.
\end{acknowledgments}

\bibliography{refs}

\appendix

\section{Proofs for the \texorpdfstring{$W_3$}{W3} and
\texorpdfstring{$M_4$}{M4} bases}
\label{app:w3-m4-proofs}

This appendix collects the two computations deferred from
\cref{sec:structured-examples}, and then records the extremal entanglement
and data-hiding properties of the $M_4$ basis quoted there.

\begin{proof}[Proof of \cref{prop:w-encoding}]
The unitary $q$ is a $\pi$-rotation about the Bloch axis
$(0,b,a)$.  A direct conjugation gives
\begin{align}
  q^\dagger Xq&=-X,\nonumber\\
  q^\dagger Zq&=cY+dZ=:\widetilde Z,\nonumber\\
  c&=\sqrt{\frac23},\qquad d=\frac1{\sqrt3}.
\end{align}
The conjugated stabilizer generators are therefore
\begin{equation}\label{eq:w-encoders}
  E_i:=V^\dagger K_iV
  =-X_i\widetilde Z_j\widetilde Z_k,
  \qquad\{i,j,k\}=\{1,2,3\}.
\end{equation}
They commute and square to the identity because they are unitary conjugates
of the triangle-graph stabilizer generators.

It remains only to check orthogonality of their orbit.  Since
$\ket{W_3}$ is permutation symmetric, the expectation of a nonidentity
product $E_1^{s_1}E_2^{s_2}E_3^{s_3}$ depends only on the Hamming weight of
$s$.  The useful products are
\begin{align}
  \widetilde Z^2&=\openone,\qquad
  E_1E_2=(cZ-dY)_1(cZ-dY)_2,\nonumber\\
  E_1E_2E_3&=X_1X_2X_3.
\end{align}
For the $W$ state,
\begin{align}
  \langle Z_1Z_2\rangle&=-\frac13,\qquad
  \langle Y_1Y_2\rangle=\frac23,\nonumber\\
  \langle Z_1Y_2\rangle&=\langle Y_1Z_2\rangle=0.
\end{align}
The first two values follow immediately by restricting to the three
one-excitation basis vectors; the mixed terms vanish because their matrix
elements are imaginary whereas $\ket{W_3}$ is real.  We obtain
\begin{align}
  \langle W_3|E_1|W_3\rangle&=0,\nonumber\\
  \langle W_3|E_1E_2|W_3\rangle
  &=-\frac{c^2}{3}+\frac{2d^2}{3}=0,\label{eq:w-three-checks}\\
  \langle W_3|E_1E_2E_3|W_3\rangle&=0,\nonumber
\end{align}
where the middle expression vanishes because $c^2=2d^2$.  The first and
third values need no computation at all: every term of $E_1$ either maps
the one-excitation subspace to a different excitation sector or has an
imaginary expectation value in the real state $\ket{W_3}$, and
$X_1X_2X_3$ maps the one-excitation subspace to the two-excitation
subspace.  Thus every nonidentity element of the eight-element group has
zero expectation value in $\ket{W_3}$, and the Gram matrix of the orbit
\eqref{eq:w-basis} is the identity.
\end{proof}

\begin{proof}[Proof of \cref{prop:m4-separator}]
By \cref{thm:commuting-characterization}, it is enough to prove that no local
change of basis makes $M_4$ flat in a graph-state basis.  We translate that
question into elementary geometry in three-dimensional space.  All
expectation values below are taken in $\ket{M_4}$.  Its one-qubit
expectation values vanish, and direct evaluation on the Pauli axes, extended
by multilinearity, gives the following formulas for distinct qubits.

For a unit vector $\mathbf a=(a_x,a_y,a_z)$, write
$\sigma_{\mathbf a}=a_xX+a_yY+a_zZ$.  A direct calculation from the six
amplitudes of $M_4$ gives, for distinct qubits,
\begin{subequations}\label{eq:m4-correlators}
\begin{align}
 \bigl\langle\sigma_{\mathbf a}^{(j)}
 \sigma_{\mathbf b}^{(k)}\bigr\rangle
   &=-\frac13\,\mathbf a\mathbin\cdot\mathbf b,
   \label{eq:m4-corr-two}\\[2pt]
 \bigl\langle\sigma_{\mathbf a}^{(j)}
 \sigma_{\mathbf b}^{(k)}
 \sigma_{\mathbf c}^{(l)}\bigr\rangle
   &=\frac{s_{jkl}}{\sqrt3}\det(\mathbf a,\mathbf b,\mathbf c),
   \label{eq:m4-corr-three}\\[2pt]
 3\bigl\langle\sigma_{\mathbf a}\otimes\sigma_{\mathbf b}
 \otimes\sigma_{\mathbf c}\otimes\sigma_{\mathbf d}\bigr\rangle
   &=(\mathbf a\mathbin\cdot\mathbf b)
     (\mathbf c\mathbin\cdot\mathbf d)
     +(\mathbf a\mathbin\cdot\mathbf c)
     (\mathbf b\mathbin\cdot\mathbf d)\nonumber\\[-2pt]
   &\quad +(\mathbf a\mathbin\cdot\mathbf d)
     (\mathbf b\mathbin\cdot\mathbf c),
   \label{eq:m4-corr-four}
\end{align}
\end{subequations}
where $s_{jkl}=\pm1$.  Its sign depends on the ordered triple of qubits but
will not matter.

After an arbitrary local change of basis, the Pauli axes at site $j$ become
an orthonormal frame
$(\mathbf x_j,\mathbf y_j,\mathbf z_j)$ in $\mathbb R^3$.  If the state
were flat in the graph basis, every nonidentity graph stabilizer would have
zero expectation value.  Indeed, each such stabilizer has equally many
$+1$ and $-1$ eigenvalues across the graph basis, while flatness assigns
equal probability to all basis vectors.  Calling the number of qubits on
which a stabilizer acts nontrivially its \emph{weight}, the weight-one
conditions already hold because all one-qubit expectations vanish.
Equations \eqref{eq:m4-corr-two}--\eqref{eq:m4-corr-four} turn this into three
concrete geometric rules: the two axes of a weight-two stabilizer must be
perpendicular; the three axes of a weight-three stabilizer must lie in one
plane; and the four axes of a weight-four stabilizer must make the
right-hand side of \eqref{eq:m4-corr-four} vanish.

We shall repeatedly use one simple observation about two orthonormal frames.
Suppose $(\mathbf a_1,\mathbf a_2,\mathbf a_3)$ and
$(\mathbf b_1,\mathbf b_2,\mathbf b_3)$ obey
$\mathbf a_i\perp\mathbf b_i$ for all $i$.  Then the three planes
$\operatorname{span}(\mathbf a_i,\mathbf b_i)$ are the three coordinate
planes of some orthonormal basis, and their common intersection is
$\{0\}$.  To check this, form the ordinary $3\times3$ matrix of dot
products $\mathbf a_i\mathbin\cdot\mathbf b_j$.  This is the change-of-basis
matrix between the two frames, so it is orthogonal; the assumed
perpendicularities make its diagonal zero.  Writing its remaining entries as
$\left(\begin{smallmatrix}0&a&b\\c&0&d\\e&f&0\end{smallmatrix}\right)$,
the scalar products of different rows give $bd=af=ce=0$.  Together with
unit row lengths, this forces exactly one entry of magnitude one in each row
and column, which gives the claim.

For a graph state, a product of single-qubit Clifford unitaries can toggle
every edge between the neighbours of one chosen vertex while leaving all
other edges unchanged.  This operation is called \emph{local
complementation} \cite{VandenNest2004,Hein2006}.  A product of
single-qubit Cliffords is itself a local unitary, so it can be absorbed
into the arbitrary local change of basis; graphs related by local
complementation therefore pose the same flatness question.

Four vertices have six possible edges, so there are $2^6=64$ labelled
graphs (eleven types up to relabelling alone).  An exact enumeration in the
companion repository divides them into six classes under local
complementation and relabelling (\cref{fig:m4-graph-classes}).  In the order below, their sizes are
$1,5,16,6,3,$ and $33$, which sum to $64$.  Writing $jk$ for the edge
$\{j,k\}$, representatives are
\begin{equation}\label{eq:m4-six-graphs}
 \begin{gathered}
 \varnothing,\quad \{12,13,14\},\quad \{12,23\},\\
 \{12\},\quad \{12,34\},\quad \{12,23,34\}.
 \end{gathered}
\end{equation}
Relabelling does not change the zero conditions in
\eqref{eq:m4-correlators}: it can only change the irrelevant sign in the
three-body formula.  In the order shown, these are the empty graph, a star,
a three-vertex path with one isolated vertex, one edge with two isolated
vertices, two disjoint edges, and a four-vertex path.

We exclude them one at a time below.  The empty graph and the star are
excluded by perpendicularity alone: at most three nonzero vectors in
$\mathbb R^3$ can be mutually perpendicular.  The next two cases also use
the coplanarity rule, and the final two require the four-body rule
\eqref{eq:m4-corr-four}.

\emph{No edges.}  The stabilizer contains $X_jX_k$ for every pair of
vertices.  Hence $\mathbf x_1,\ldots,\mathbf x_4$ would be four mutually
perpendicular unit vectors in $\mathbb R^3$, which is impossible.

\emph{A star.}  Let vertex 1 be the centre.  Products of pairs of leaf
generators make $\mathbf x_2,\mathbf x_3,\mathbf x_4$ mutually
perpendicular.  The three leaf generators make $\mathbf z_1$ perpendicular
to each of them.  Again, four nonzero vectors cannot be mutually
perpendicular in $\mathbb R^3$.

\emph{A three-vertex path and one isolated vertex.}  For the path
$1-2-3$, the three weight-two stabilizers
$X_1Z_2$, $Z_2X_3$, and $X_1X_3$ make
$\mathbf x_1,\mathbf z_2,\mathbf x_3$ mutually perpendicular.  Multiplying
those stabilizers by the isolated generator $X_4$ gives three weight-three
stabilizers.  They require $\mathbf x_4$ to lie in each of the three
coordinate planes determined by
$\mathbf x_1,\mathbf z_2,\mathbf x_3$.  Their intersection is $\{0\}$,
contrary to $\lVert\mathbf x_4\rVert=1$.

\emph{One edge and two isolated vertices.}  For edge $12$, the three
nonidentity stabilizers carry the local labels $XZ,ZX,YY$.  Thus
\begin{equation}
 \mathbf x_1\perp\mathbf z_2,\qquad
 \mathbf z_1\perp\mathbf x_2,\qquad
 \mathbf y_1\perp\mathbf y_2.
\end{equation}
Our observation says that their three spanning planes have only the zero
vector in common.  Multiplying the three edge stabilizers by the isolated
generator $X_3$, however, requires $\mathbf x_3$ to lie in all three
planes.  This is impossible for a unit vector.

\emph{Two disjoint edges.}  For edge $12$, set
\begin{equation}
 (\mathbf a_i,\mathbf b_i)_{i=1}^{3}
 =\bigl((\mathbf x_1,\mathbf z_2),
        (\mathbf z_1,\mathbf x_2),
        (\mathbf y_1,\mathbf y_2)\bigr),
\end{equation}
and define the real matrices
\begin{equation}
 S_i=\mathbf a_i\mathbf b_i^{T}+\mathbf b_i\mathbf a_i^{T}.
\end{equation}
In the coordinates supplied by the observation above, the three $S_i$
occupy the three different off-diagonal pairs.  They are therefore
independent symmetric matrices with zero trace.  For edge $34$, define
\begin{equation}
 (\mathbf c_j,\mathbf d_j)_{j=1}^{3}
 =\bigl((\mathbf x_3,\mathbf z_4),
        (\mathbf z_3,\mathbf x_4),
        (\mathbf y_3,\mathbf y_4)\bigr)
\end{equation}
and
\begin{equation}
 T_j=\mathbf c_j\mathbf d_j^{T}+\mathbf d_j\mathbf c_j^{T}.
\end{equation}
The same argument shows that $T_1,T_2,T_3$ are independent.  The product
of any nonidentity stabilizer of the first edge with one of the second edge
is a weight-four stabilizer.  Since the two vectors within each pair are perpendicular,
\eqref{eq:m4-corr-four} says that its expectation vanishes exactly when
\begin{equation}
 \sum_{r,s=1}^{3}(S_i)_{rs}(T_j)_{rs}=0.
 \label{eq:m4-matrix-orthogonality}
\end{equation}
Thus every one of the three independent $S_i$'s would be orthogonal to
every one of the three independent $T_j$'s under the entrywise inner
product in \eqref{eq:m4-matrix-orthogonality}.  But all these matrices live in the five-dimensional space of
real symmetric $3\times3$ matrices with zero trace: symmetry leaves six
entries and the trace removes one.  The vectors perpendicular to any three
independent ones in this five-dimensional space span at most two
dimensions, so they cannot contain the three independent $T_j$'s.

\emph{A four-vertex path.}  Finally take $1-2-3-4$.  The two weight-two
stabilizers give
\begin{equation}
 \mathbf x_1\perp\mathbf z_2,\qquad
 \mathbf z_3\perp\mathbf x_4.
 \label{eq:m4-path-pairs}
\end{equation}
Let $P=\operatorname{span}(\mathbf x_1,\mathbf z_2)$, and let $\mathbf n$
be its unit normal.  Four weight-three stabilizers have axis triples
\begin{equation}\label{eq:m4-path-first-triples}
 \begin{gathered}
 (\mathbf x_1,\mathbf x_3,\mathbf z_4),\quad
 (\mathbf z_2,\mathbf x_3,\mathbf z_4),\\
 (\mathbf x_1,\mathbf y_3,\mathbf y_4),\quad
 (\mathbf z_2,\mathbf y_3,\mathbf y_4).
 \end{gathered}
\end{equation}
The first two coplanarity conditions say that
$\mathbf x_3\times\mathbf z_4$ is perpendicular to both spanning vectors
of $P$, and hence parallel to $\mathbf n$; the last two say the same for
$\mathbf y_3\times\mathbf y_4$.  Thus each of the pairs
$(\mathbf x_3,\mathbf z_4)$ and $(\mathbf y_3,\mathbf y_4)$ is either
parallel (its cross product is zero) or lies in $P$ (its nonzero cross
product is normal to $P$).  They cannot both lie in $P$: then
$\mathbf z_3\parallel\mathbf x_3\times\mathbf y_3$ and
$\mathbf x_4\parallel\mathbf y_4\times\mathbf z_4$ would both be normal to
$P$, contradicting $\mathbf z_3\perp\mathbf x_4$.  They cannot both be
parallel either: the planes spanned by $(\mathbf x_3,\mathbf y_3)$ and
$(\mathbf z_4,\mathbf y_4)$ would coincide, making their normals
$\mathbf z_3$ and $\mathbf x_4$ parallel.  Hence exactly one pair is
parallel and the other spans $P$.  Both vectors in the parallel pair are
perpendicular to both vectors of the planar pair, so they are parallel to
$\mathbf n$.  (Say $\mathbf x_3\parallel\mathbf z_4$.  Then
$\mathbf x_3\perp\mathbf y_3$ within the frame at site 3, and
$\mathbf x_3\parallel\mathbf z_4\perp\mathbf y_4$ transfers the same
conclusion to site 4; the other case is identical.)  The remaining axes $\mathbf z_3,\mathbf x_4$ lie in
$P$; each is a right-angle rotation of the corresponding vector in the
planar pair.  Since $\mathbf z_3\perp\mathbf x_4$, both planar pairs are
orthonormal bases of $P$.

The other four weight-three stabilizers have triples
\begin{equation}\label{eq:m4-path-second-triples}
 \begin{gathered}
 (\mathbf y_1,\mathbf y_2,\mathbf z_3),\quad
 (\mathbf y_1,\mathbf y_2,\mathbf x_4),\\
 (\mathbf z_1,\mathbf x_2,\mathbf z_3),\quad
 (\mathbf z_1,\mathbf x_2,\mathbf x_4).
 \end{gathered}
\end{equation}
The same reasoning, now using the plane spanned by
$\mathbf z_3,\mathbf x_4$, shows that both vectors in one of the pairs
$(\mathbf y_1,\mathbf y_2)$ and $(\mathbf z_1,\mathbf x_2)$ are parallel
to $\mathbf n$, while the other pair is an orthonormal basis of $P$.
Consequently, exactly one of $\mathbf y_1,\mathbf x_2$ is parallel to
$\mathbf n$; call the other vector $\mathbf u\in P$.

Consider finally two of the weight-four stabilizers, $YXXY$ and $YXYZ$.
Both contain the same two axes at sites 1 and 2: one is parallel to
$\mathbf n$, and the other is the unit vector $\mathbf u\in P$.  At
sites 3 and 4, each stabilizer contains one further axis parallel to
$\mathbf n$.  The two remaining planar axes, one from each stabilizer, are
the vectors $\mathbf e_1,\mathbf e_2$ of the orthonormal basis of $P$
found above.  The four-body formula therefore gives, up to signs,
\[
  \langle YXXY\rangle=\pm\frac{\mathbf u\mathbin\cdot\mathbf e_1}{3},
  \qquad
  \langle YXYZ\rangle=\pm\frac{\mathbf u\mathbin\cdot\mathbf e_2}{3}.
\]
Because $\mathbf u$ is a unit vector in $P$, its squared components in
the orthonormal basis $(\mathbf e_1,\mathbf e_2)$ sum to one.  Therefore
\begin{equation}\label{eq:m4-path-final}
 \langle YXXY\rangle^2+\langle YXYZ\rangle^2=\frac19.
\end{equation}
In particular, the two expectations cannot both vanish.

Every graph in \eqref{eq:m4-six-graphs} is therefore excluded.  No local
change of basis can make $M_4$ flat in a graph basis, so
$\ket{M_4}\notin\ECI^{(4)}$.
\end{proof}

\paragraph{Maximal entanglement and local discrimination.}
Let
\begin{equation}\label{eq:product-overlap}
  \Lambda^2(\psi)=
  \max_{\ket{a_1}\cdots\ket{a_4}}
  \left|\braket{a_1a_2a_3a_4|\psi}\right|^2
\end{equation}
be the maximal squared product overlap underlying the geometric measure
\cite{WeiGoldbart2003}, where the maximum is over normalized one-qubit
states.  Every four-qubit pure state obeys $\Lambda^2\ge2/9$.  Choose the
first-qubit Schmidt vector whose squared coefficient is at least $1/2$.
The corresponding normalized three-qubit state has squared product overlap
at least $4/9$ \cite{Tamaryan2009,Chen2010}; multiplying the two bounds
gives $2/9$.  A smaller value of $\Lambda^2$ means larger geometric
entanglement.
The state $\ket{M_4}$ attains this bound.  Grouping its terms according to
the fourth qubit gives
\begin{equation}\label{eq:m4-chiral-decomposition}
\begin{aligned}
  \ket{M_4}&=\frac1{\sqrt2}
  \left(\ket{\phi_1}\ket1+\ket{\phi_2}\ket0\right),\\
  \ket{\phi_1}&=\frac1{\sqrt3}
  \left(\ket{001}+\omega\ket{010}+\omega^2\ket{100}\right),\\
  \ket{\phi_2}&=\frac1{\sqrt3}
  \left(\ket{110}+\omega\ket{101}+\omega^2\ket{011}\right).
\end{aligned}
\end{equation}
The vectors $\ket{\phi_1},\ket{\phi_2}$ span the three-qubit chiral
subspace.  It was proven that every normalized state in
this subspace has maximal squared product overlap $4/9$
\cite{DenkerImaiGuhne2025}.  Contracting the fourth qubit of $\ket{M_4}$
with any normalized state therefore leaves, up to the factor $1/\sqrt2$,
a normalized vector in this subspace.  Hence
$\Lambda^2(M_4)=(1/2)(4/9)=2/9$.  Together with the universal lower bound
above, this shows that $\ket{M_4}$ globally maximizes the geometric measure
of four-qubit entanglement.  Local-unitary invariance gives the same
conclusion for every vector in the basis of \cref{prop:m4-encoding}.


The original Higuchi--Sudbery criterion gives a related, but more nuanced,
extremal statement.  A direct partial trace shows that every two-qubit
reduction of $\ket{M_4}$ has spectrum
$\{1/2,1/6,1/6,1/6\}$, and hence von Neumann entropy
\begin{equation}\label{eq:m4-entropy}
  S(\rho_{jk})=1+\frac12\log_2 3.
\end{equation}
Higuchi and Sudbery proposed maximizing the average of this entropy over the
three inequivalent $2|2$ cuts.  The state is a proven local maximum, and
Gour and Wallach proved global maximality within the 1-uniform family; its
unrestricted global maximality remains conjectural
\cite{Higuchi2000,Brierley2007,Gour2010}.  By contrast, the associated
average linear entropy is globally maximized by $\ket{M_4}$: using the
definition without an additional normalization factor, each cut has
$1-\Tr(\rho_{jk}^2)=2/3$.  This global optimizer is not unique
\cite{Gour2010}.

Finally, write $\ket{M_r}=U_r\ket{M_4}$.  If $E_r$ is a fully separable
positive operator, write it as a positive sum of projectors onto normalized
product vectors.  Each squared overlap is at most $2/9$, so
\eqref{eq:product-overlap} gives
\[
  \bra{M_r}E_r\ket{M_r}\le\frac29\Tr(E_r).
\]
For any fully separable POVM $\{E_r\}_{r=1}^{16}$, averaging a uniform
label and using $\sum_rE_r=\openone$ therefore yields
\begin{equation}\label{eq:m4-sep-bound}
  P_{\mathrm{succ}}^{\mathrm{SEP}}
  =\frac1{16}\sum_r\bra{M_r}E_r\ket{M_r}
  \le\frac29.
\end{equation}
Every LOCC POVM is separable, so the same bound applies to LOCC measurements,
while the global projective measurement in the $M_4$ basis succeeds with
probability one
\cite{DiVincenzo2002,Hayashi2006}.

\section{Why the smooth-family bound applies globally}
\label{app:technical-dimension}

The main text bounds the number of state parameters on one smooth family for
which the state and a complete encoder table vary together.  Here we justify
applying that single-family bound to the full encodable set.  This extra step is
needed because a state can admit several unrelated encoder tables, and a
single choice of table need not vary smoothly across every encodable state.

Keep the state and all $N$ product encoders as simultaneous variables.
Represent the state ray by its rank-one projector and each one-qubit
$SU(2)$ factor as
$aI-i(bX+cY+dZ)$, where $a^2+b^2+c^2+d^2=1$.  Restricting from $U(2)$ to
$SU(2)$ loses nothing: the omitted phase merely multiplies one encoded basis
vector by an irrelevant global phase.  Purity, 1-uniformity, the unitary
normalizations, and the real and imaginary parts of all ONB conditions are
then polynomial equations in real variables.  Regularity says that at least
one determinant from a finite list is nonzero, so it is also a polynomial
condition (an inequality rather than an equality).

Sets described in this way, by polynomial equations and inequalities, are
called \emph{semialgebraic}.  A standard structure
result\footnote{Theorem~9.1.8 of Ref.~\cite{Bochnak1998} states
that a semialgebraic set can be divided into finitely many smooth pieces.
We divide the pieces further, if necessary, so that the rank of the map
which retains only the state is constant on each piece.}
therefore divides the simultaneous state-and-encoder solutions into finitely
many smooth pieces.  On each piece, the number of independently varying
state parameters is the rank of the projection that forgets the encoders.
Every smooth family of rank-one projectors has, locally, a smooth normalized
ket representative.  Therefore any two state directions in a projected
piece lift to two parameter changes of a ket together with its complete
encoder table.  The equal-contributions identity,
\cref{prop:equal-contributions}, applies to these lifted changes, and the
argument in \cref{sec:dimension-count} bounds the projection rank by $N-1$.

Finally, forgetting the encoders projects the full solution set onto the set
of encodable regular 1-uniform states.  It is a finite union of the projected
smooth pieces, each of dimension at most $N-1$.  The full encodable set
therefore has dimension at most $N-1$, which is precisely
\eqref{eq:global-bound}.

\section{Regularity and nonencodability of
\texorpdfstring{$\ket{\Psi_{\mathrm{ex}}}$}{Psi\_ex}}
\label{app:explicit-state}

We prove \cref{cor:explicit-nonencodable} by placing its state in a simple
eight-parameter family.  For $r>0$ and $|u_j|=1$, let
\begin{align}\label{eq:explicit-family}
 \ket{\Psi(r,\bm u)}\propto{}&
 \ket{0000}+u_1\ket{1111}\nonumber\\[-2pt]
 &+u_2\ket{0011}+u_3\ket{1100}\nonumber\\[-2pt]
 &+u_4\ket{0101}+u_5\ket{1010}\nonumber\\[-2pt]
 &+r u_6\ket{0110}+r u_7\ket{1001}.
\end{align}
Fixing the first coefficient and requiring $r>0$ removes the overall scale
and phase; the eight parameters are then uniquely determined by the ray.
Every member is 1-uniform: even parity
makes all one-qubit off-diagonal entries vanish, while equal magnitudes
within each complementary pair balance the two populations.

\emph{A criterion for regularity.}
Regularity means that the twelve marginal conditions can change
independently to first order.  This can be checked using only ordinary inner
products.  For real numbers $r_\mu$, set
$R=\sum_\mu r_\mu P_\mu$.  The same linear combination of the twelve
first-order marginal changes is
\[
  \delta_v\langle R\rangle_\psi
  =2\operatorname{Re}\langle R\psi|v\rangle.
\]
If this expression vanished for every phase-fixed change $v$, we could choose
$v=R\ket\psi$: it is phase-fixed because
$\langle R\rangle_\psi=0$ for a 1-uniform state.  We would obtain
$2\|R\ket\psi\|^2=0$.  Thus the twelve marginal conditions are independent
exactly when the map $R\mapsto R\ket\psi$ has no nonzero kernel.  It is
therefore enough to write this map as a real matrix and exhibit one nonzero
$12\times12$ determinant formed from its rows.

\emph{Verifying the criterion on a slice.}
To show that such a determinant is not always zero in
\eqref{eq:explicit-family}, set all $u_j=1$ and consider the simpler states
whose two amplitudes within each complementary pair are equal and real,
\begin{equation}\label{eq:paired-slice}
  \begin{aligned}
  \ket{\Psi_{ABCD}}\propto{}&
  A(\ket{0000}+\ket{1111})
  +B(\ket{0011}+\ket{1100})\\[-2pt]
  &+C(\ket{0101}+\ket{1010})
  +D(\ket{0110}+\ket{1001}),
  \end{aligned}
\end{equation}
reached at $(A,B,C,D)=(1,1,1,r)$; for this calculation we allow arbitrary
real $A,B,C,D$.  Parity separates the changes generated by local $Z$
rotations from those generated by local $X$ and $Y$ rotations.  An exact
coefficient calculation gives the following selected $12\times12$ real
determinant of the matrix representing
$R\mapsto R\ket{\Psi_{ABCD}}$.  The selected rows are the real parts of
the output amplitudes on
$0000,0001,0010,0011,0100,0101,0110,0111$, together with the imaginary
parts on $0001,0010,0100,0111$:
\begin{equation}\label{eq:regular-minor}
  16ABCD
  \prod_{\epsilon_B,\epsilon_C,\epsilon_D\in\{\pm1\}}
  (A+\epsilon_BB+\epsilon_CC+\epsilon_DD).
\end{equation}
A single nonzero $12\times12$ determinant is enough to prove independence.
At $(A,B,C,D)=(1,1,1,2)$,
which is the member of \eqref{eq:explicit-family} with all $u_j=1$ and
$r=2$, it equals $-4320$.  The same determinant is a polynomial in the real
coordinates of the full family \eqref{eq:explicit-family}, and we have just
exhibited a point of that family at which it does not vanish; it is
therefore not identically zero on \eqref{eq:explicit-family}.

\emph{Counting encodable parameters inside the family.}
We can obtain the required seven-parameter bound directly, without splitting
the changes into local and nonlocal ones.  Normalize the state in
\eqref{eq:explicit-family} and write
\[
  \mathcal N^2=6+2r^2,\qquad
  \ket\psi=\ket{\Psi(r,\bm u)},
\]
so that at $r=\mathrm e$ the normalization $\mathcal N$ agrees with the one
in \cref{cor:explicit-nonencodable}.
Consider two ordinary parameter changes.  The first changes the magnitude
$r$ of the final complementary pair; the second changes the phase
$u_1=e^{i\theta_1}$ of the $\ket{1111}$ coefficient.  After removing the
irrelevant first-order global phase, representatives of their ray
derivatives are
\begin{align}
  v=\partial_r\psi
  &=\frac{u_6\ket{0110}+u_7\ket{1001}}{\mathcal N}
    -\frac{2r}{\mathcal N^2}\ket\psi,\label{eq:explicit-r-derivative}\\
  w=(I-\ketbra{\psi}{\psi})\partial_{\theta_1}\psi
  &=\frac{iu_1\ket{1111}}{\mathcal N}
    -\frac{i}{\mathcal N^2}\ket\psi.\label{eq:explicit-phase-derivative}
\end{align}
The projector in the second line simply removes the infinitesimal global
phase: $\langle\psi|\partial_{\theta_1}\psi\rangle=i/\mathcal N^2$.
Directly from the eight orthogonal computational-basis kets,
\begin{align}\label{eq:explicit-inner-product}
  \langle v|w\rangle
    &=-\frac{2ir}{\mathcal N^4},\nonumber\\
  \omega_\psi(v,w)
    &=2\,\operatorname{Im}\langle v|w\rangle
      =-\frac{4r}{\mathcal N^4}\nonumber\\
    &=-\frac{r}{(r^2+3)^2}\ne0.
\end{align}
Both parameter changes remain inside the 1-uniform family
\eqref{eq:explicit-family}, and the area \eqref{eq:explicit-inner-product}
is nonzero at every member.  Now suppose that one of the smooth pieces of
the state-and-encoder solution set, restricted to this family, reached all
eight state parameters.  Near some member, the state directions obtained
by forgetting the encoders would then fill the eight-dimensional tangent
space of the family; in particular, they would include both $v$ and $w$.
The equal-contributions identity would require $\omega_\psi(v,w)=0$,
contradicting \eqref{eq:explicit-inner-product}.  Thus every projected
smooth piece has dimension at most seven.

It remains to turn this local statement into an equation that the explicit
choice of parameters can avoid.  As explained in
\cref{app:technical-dimension}, the state-and-encoder conditions are
polynomial equalities and inequalities with rational coefficients.  Add the
parameter $r$ and the real and imaginary parts of the seven $u_j$ to that
description, with the seven circle equations
$(\operatorname{Re}u_j)^2+(\operatorname{Im}u_j)^2=1$.  Removing the
auxiliary encoder variables leaves a semialgebraic set of encodable
parameters; Proposition~2.2.4 of Ref.~\cite{Bochnak1998} justifies this
projection.  The finite smooth-piece argument of
\cref{app:technical-dimension}, now restricted to
\eqref{eq:explicit-family}, shows that their union---the entire projected
encodable parameter set---has dimension at most seven, whereas the full
family has dimension eight.  Proposition~2.8.2 of the same reference says
that a semialgebraic set and its smallest algebraic closure have the same
dimension.  Consequently, one additional nonzero polynomial, not identically
zero on the full family, vanishes on every encodable member.  Since the
defining equations have rational coefficients, this polynomial may be chosen
with algebraic coefficients.

\emph{Why the chosen phases avoid every such equation.}
For the state in \eqref{eq:explicit-nonencodable},
\[
  r=\mathrm e,\qquad u_j=\mathrm e^{i\alpha^j},\qquad
  \alpha=2^{1/7}.
\]
The algebraic numbers $1,i\alpha,i\alpha^2,\ldots,i\alpha^7$ are linearly
independent over the rationals.  Indeed, real and imaginary parts first
separate the initial $1$, and dividing any relation among
$\alpha,\ldots,\alpha^7$ by $\alpha$ would give a polynomial of degree
at most six satisfied by the degree-seven number $\alpha$.
The Lindemann--Weierstrass theorem says that exponentials of distinct
algebraic numbers are linearly independent over the algebraic numbers.  A
polynomial relation among the following eight numbers would expand into
such a linear relation; the rational independence just proved ensures that
all exponent sums are distinct.  Therefore
\begin{equation}\label{eq:explicit-independent}
  \mathrm e,\ \mathrm e^{i\alpha},\ \mathrm e^{i\alpha^2},\ldots,
  \mathrm e^{i\alpha^7}
\end{equation}
obey no nonzero polynomial equation with algebraic coefficients by the
Lindemann--Weierstrass theorem (Ref.~\cite{Popescu2024}, Thm.~3.2).

Writing $u_j=x_j+i y_j$, the circle condition gives
$x_j=(u_j+u_j^{-1})/2$ and $y_j=(u_j-u_j^{-1})/(2i)$.  A polynomial
equation in the real coordinates that held at the displayed state would
therefore become, after multiplying by suitable powers of the $u_j$,
a polynomial equation among the eight numbers in
\eqref{eq:explicit-independent}; and if the original polynomial is not
identically zero on the family, the resulting polynomial is nonzero.  No
such equation exists.

The two remaining claims now follow.  First, the selected determinant
used to test regularity is not identically zero on the
family, so it cannot vanish at the displayed state.  The state is
therefore regular.
Second, every encodable member obeys the additional polynomial equation
found above, which is not identically zero on the family; the displayed
state can obey no such equation, so it is not encodable.  Hence
$\ket{\Psi_{\mathrm{ex}}}$ is regular and nonencodable, as claimed.  The
factorization \eqref{eq:regular-minor} and its nonvanishing at
$(1,1,1,2)$ are checked exactly in the companion repository.

\section{Exact encoding of a conjectured nonencodable candidate}
\label{app:eta-encoding}

Ref.~\cite{Pimpel2023} reported a failed numerical search for a local
encoding of the following four-qubit state:
\begin{equation}\label{eq:eta-state}
  \ket{\eta}
  =\sqrt{\frac23}\ket{W_4}
   +\frac1{\sqrt3}\ket{\mathrm{GHZ}_4}
  =\frac1{\sqrt6}\sum_{x\in S}\ket{x},
\end{equation}
where
\[
  S=\{0000,1111,0001,0010,0100,1000\}\subset\mathbb F_2^4.
\]
Here $\mathbb F_2^4$ denotes four-bit strings, with addition performed
bitwise modulo two.
The following exact construction shows that this particular candidate is
locally encodable.  Interestingly, the encoding below was proposed by ChatGPT, after extensive numerical searches had failed to
find one.

\begin{proposition}[An exact encoding of $\ket{\eta}$]
\label{prop:eta-encoding}
For $t\in\mathbb F_2^4$, write
$\boldsymbol\theta_t=(\theta_{t,1},\ldots,\theta_{t,4})$ and define
\[
  \begin{aligned}
    X^t&=\bigotimes_{j=1}^4X_j^{t_j},\\
    D_t&=\bigotimes_{j=1}^4
    \operatorname{diag}\!\left(1,e^{i\pi\theta_{t,j}}\right),\\
    U_t&=D_0^\dagger D_tX^t,
  \end{aligned}
\]
with the phases, in units of $\pi$, given by
\begin{center}
\begingroup
\footnotesize\upshape
\setlength{\tabcolsep}{2.5pt}
\renewcommand{\arraystretch}{1.04}
\begin{tabular}{@{}c|cccc@{\qquad}c|cccc@{}}
$t$&$\theta_{t,1}$&$\theta_{t,2}$&$\theta_{t,3}$&$\theta_{t,4}$&
$t$&$\theta_{t,1}$&$\theta_{t,2}$&$\theta_{t,3}$&$\theta_{t,4}$\\
\hline
0000&1/4&3/4&3/4&5/4&1000&5/4&0&0&1\\
0001&1&1/2&3/2&1/4&1001&0&1/4&5/4&0\\
0010&1/2&3/2&7/4&1&1010&3/2&1/4&1&5/4\\
0011&7/4&3/4&1/2&0&1011&3/4&1&1/4&1/4\\
0100&1/2&7/4&1/2&0&1100&3/2&1&5/4&1/4\\
0101&7/4&3/2&3/4&1&1101&3/4&5/4&1&5/4\\
0110&1/4&1/2&3/2&5/4&1110&5/4&5/4&1/4&1\\
0111&1&7/4&7/4&1/4&1111&0&0&0&0
\end{tabular}
\endgroup
\end{center}
The phase entries are understood modulo $2$.  Here $D_0$ is the diagonal
unitary specified by the $t=0000$ row; the common factor $D_0^\dagger$
ensures that $U_0=\openone$.
Then $\{U_t\ket{\eta}:t\in\mathbb F_2^4\}$ is an orthonormal basis.
\end{proposition}

\begin{proof}
The only property of $S$ needed below is that every nonzero binary shift
overlaps $S$ in exactly two strings.  Directly checking the fifteen
nonzero shifts gives
\[
  |S\cap(S+d)|=2.
\]
Because $D_t$ is diagonal and $X^t$ flips the bits indicated by $t$,
the support of $U_t\ket{\eta}$ is $S+t$.  Hence, for $t\ne u$, the two
encoded states overlap on exactly two computational strings,
$\{y_0,y_1\}=(S+t)\cap(S+u)$, and
\[
  \bra{\eta}U_t^\dagger U_u\ket{\eta}
  =\frac16\sum_{k=0}^1
    e^{i\pi(\boldsymbol\theta_u-\boldsymbol\theta_t)\cdot y_k}.
\]
For the dot product, regard $y_0,y_1$ as vectors of the ordinary integers
$0$ and $1$.  Direct substitution of the phase table gives
\[
  (\boldsymbol\theta_u-\boldsymbol\theta_t)\cdot(y_1-y_0)
  \equiv1\pmod 2
\]
for every $t\ne u$.  The ratio of the two phase factors is therefore
$e^{i\pi}=-1$, so they cancel.  Each $U_t$ is a product of one-qubit
unitaries, so all sixteen encoded states are normalized.  Sixteen pairwise
orthogonal states in a sixteen-dimensional Hilbert space form an orthonormal
basis.  The symbolic verification in the companion repository also checks
all $120$ unordered pairs exactly.
\end{proof}

\end{document}

%% file: figures/encoding_hierarchy.tex
\begin{figure}[t!]
\centering
\begin{tikzpicture}[x=1cm,y=1cm]
  \draw[rounded corners=7pt,draw=black!30,fill=black!2,line width=0.6pt]
    (0,0) rectangle (8.3,6.2);
  \node[font=\footnotesize,text=black!55,anchor=west] at (0.3,5.86)
    {$\mathbb{CP}^{15}$: four-qubit pure states};

  \draw[draw=encoderorange,fill=softorange,line width=0.9pt]
    (3.4,2.75) ellipse (3.0 and 2.1);
  \draw[draw=encoderteal,fill=softteal,line width=0.9pt]
    (2.75,2.5) ellipse (1.95 and 1.4);
  \draw[draw=encoderblue,fill=softblue,line width=0.9pt]
    (2.2,2.3) ellipse (1.0 and 0.75);

  \node[font=\scriptsize\bfseries,text=encoderblue] at (2.2,2.5)
    {$\EI=\mathrm{LME}$};
  \node[font=\small\bfseries,text=encoderteal]   at (1.8,3.3)  {$\ECI$};
  \node[font=\small\bfseries,text=encoderorange] at (2.3,4.25) {$\EG$};

  \draw[dashed,draw=black!60,line width=0.7pt,fill=black,fill opacity=0.09]
    plot[smooth cycle,tension=0.75] coordinates {
      (1.6,1.95) (2.0,2.15) (2.6,2.1) (3.2,2.25) (3.8,2.15) (4.4,2.35)
      (4.95,2.6) (5.55,2.85) (6.15,2.85) (6.35,3.6) (6.55,4.4) (7.1,4.8)
      (7.75,4.6) (8.0,3.9) (7.9,3.1) (7.4,2.4) (6.8,1.8) (6.1,1.65)
      (5.5,1.65) (4.9,1.75) (4.3,1.7) (3.7,1.8) (3.1,1.7) (2.5,1.7)
      (2.0,1.75) };
  \node[font=\small\bfseries,text=black!60,align=center] at (7.2,4.05)
    {$\mathcal Z_4$\\[-1pt]{\scriptsize\mdseries 1-uniform}};

  \fill[encoderteal] (4.1,3.15) circle (2.1pt);
  \node[font=\scriptsize,anchor=east] at (4.0,3.15)
    {$\ket{W_3}\!\otimes\!\ket0$};
  \fill[encoderorange!90!black] (5.4,2.25) circle (2.1pt);
  \node[font=\scriptsize,anchor=west] at (5.52,2.25) {$\ket{M_4}$};
  \fill[black!80] (7.15,3.35) circle (2.1pt);
  \node[font=\scriptsize,anchor=north] at (7.15,3.22)
    {$\ket{\Psi_{\mathrm{ex}}}$};

  \node[font=\scriptsize,text=black!70] at (2.55,0.33)
    {$\EI^{(4)}\subsetneq\ECI^{(4)}\subsetneq\EG^{(4)}
      \subsetneq\mathbb{CP}^{15}$};
\end{tikzpicture}
\caption{\label{fig:state-class-map}%
The four-qubit encoding hierarchy
\eqref{eq:strict-four-qubit-hierarchy}.  Every inclusion is strict:
$\ket{W_3}\otimes\ket0$ is reached by commuting reflections but is not
LME (\cref{rem:four-qubit-w-witness}); $\ket{M_4}$ is locally encodable
but by no commuting reflections
(\cref{prop:m4-encoding,prop:m4-separator});
$\ket{\Psi_{\mathrm{ex}}}$ admits no local encoding at all
(\cref{cor:explicit-nonencodable}).  The dashed region is the 1-uniform
family $\mathcal Z_4$ defined in \eqref{eq:uniform-family}.  It contains
$\ket{M_4}$ and $\ket{\Psi_{\mathrm{ex}}}$ and reaches into $\EI$, for
instance at the GHZ and cluster states.  At its regular points it has
dimension $18$, whereas its encodable regular part has dimension at most
$15$ (\cref{thm:nonencodable}); all boundaries are schematic.}
\end{figure}

%% file: figures/orbit_concept.tex
\begin{figure}[t]
\centering
\begin{tikzpicture}[x=1cm,y=1cm]
  \draw[rounded corners=7pt,draw=black!25,fill=black!2,line width=0.6pt]
    (0,0) rectangle (8.3,4.55);
  \node[font=\scriptsize,text=black!50,anchor=east] at (8.12,0.36)
    {pure states $\mathbb{CP}^{N-1}$};

  \definecolor{orbitviolet}{RGB}{124,88,160}
  \draw[encoderteal!60,line width=1.15pt]
    plot [smooth cycle,tension=0.8] coordinates
    {(4.9,2.3) (4.45,3.18) (2.75,3.52) (1.45,2.95) (0.9,2.3)
     (1.62,1.55) (3.05,1.12) (4.3,1.48)};
  \node[font=\scriptsize,text=encoderteal!80!black] at (2.55,0.72)
    {local-unitary orbit};

  \draw[fill=encoderblue,draw=encoderblue!55!black,line width=0.35pt]
    (0.9,2.3) circle (2.4pt);
  \node[font=\small,anchor=east,text=encoderblue] at (0.68,2.3)
    {$\ket\psi$};
  \draw[fill=encoderorange,draw=encoderorange!55!black,line width=0.35pt]
    (2.75,3.52) circle (2.4pt);
  \node[font=\small,anchor=south,text=encoderorange] at (2.75,3.68)
    {$U_a\ket\psi$};
  \draw[fill=encoderteal,draw=encoderteal!55!black,line width=0.35pt]
    (4.9,2.3) circle (2.4pt);
  \draw[fill=orbitviolet,draw=orbitviolet!55!black,line width=0.35pt]
    (3.05,1.12) circle (2.4pt);

  \draw[-{Latex[length=1.8mm]},black!45,line width=0.7pt]
    (5.15,2.42) -- (5.9,2.42);
  \foreach \i in {0,1,2,3}{
    \foreach \j in {0,1,2,3}{
      \draw[black!25,line width=0.3pt]
        ({6.05+0.26*\i},{1.72+0.26*\j}) rectangle
        ({6.31+0.26*\i},{1.98+0.26*\j});
    }
  }
  \fill[encoderblue!85]  (6.05,2.50) rectangle (6.31,2.76);
  \fill[encoderorange!85] (6.31,2.24) rectangle (6.57,2.50);
  \fill[encoderteal!85]  (6.57,1.98) rectangle (6.83,2.24);
  \fill[orbitviolet!85]  (6.83,1.72) rectangle (7.09,1.98);
  \node[font=\scriptsize,text=black!65] at (6.57,1.42)
    {$\langle\phi_a|\phi_b\rangle=\delta_{ab}$};
\end{tikzpicture}
\caption{\label{fig:orbit-concept}%
The state-first question as geometry.  Product unitaries move the
fiducial state $\ket\psi$ along its local-unitary orbit (curve); a local
encoding (\cref{def:local-encoding}) exists precisely when the orbit
contains $N=2^n$ mutually orthogonal states (dots, shown for $N=4$;
each dot is colored as its diagonal entry of the Gram matrix).
These states then form an iso-entangled basis: a complete projective
measurement all of whose outcomes carry exactly the entanglement of
$\ket\psi$.}
\end{figure}
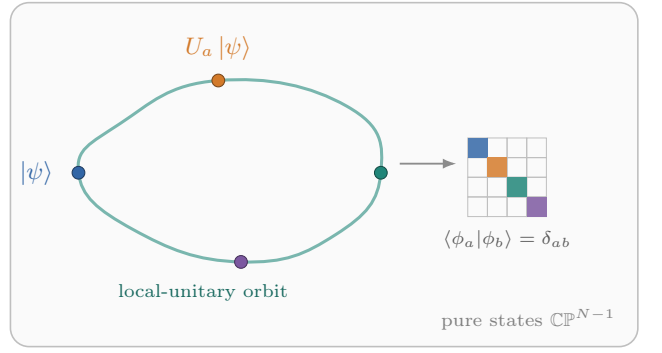

%% file: figures/encoder_coordination.tex
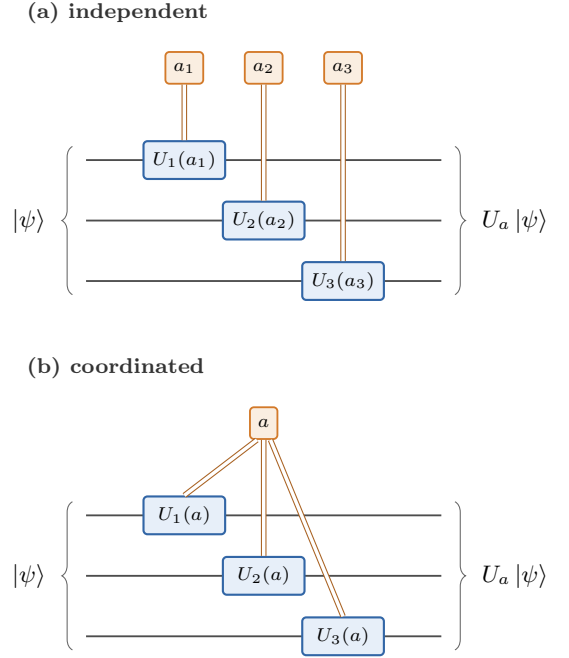
\begin{figure}[t]
\centering
\begin{tikzpicture}[x=1cm,y=1cm,
  qwire/.style={black!70,line width=0.7pt},
  cwire/.style={encoderorange!80!black,double,double distance=1.3pt,
                line width=0.4pt},
  gate/.style={draw=encoderblue,fill=softblue,line width=0.8pt,
               rounded corners=1.5pt,minimum width=1.08cm,
               minimum height=0.5cm,font=\scriptsize},
  cbit/.style={draw=encoderorange,fill=softorange,line width=0.7pt,
               rounded corners=1.5pt,minimum height=0.42cm,
               inner xsep=3pt,font=\scriptsize},
  ptitle/.style={font=\footnotesize\bfseries,text=black!80}]
  \begin{scope}
    \node[ptitle,anchor=west] at (-0.9,1.95) {(a) independent};
    \node[cbit] (a1) at (1.3,1.22)  {$a_1$};
    \node[cbit] (a2) at (2.35,1.22) {$a_2$};
    \node[cbit] (a3) at (3.4,1.22)  {$a_3$};
    \foreach \y in {0,-0.8,-1.6}
      \draw[qwire] (0,\y) -- (4.7,\y);
    \node[gate] (g1) at (1.3,0)     {$U_1(a_1)$};
    \node[gate] (g2) at (2.35,-0.8) {$U_2(a_2)$};
    \node[gate] (g3) at (3.4,-1.6)  {$U_3(a_3)$};
    \draw[cwire] (a1.south) -- (g1.north);
    \draw[cwire] (a2.south) -- (g2.north);
    \draw[cwire] (a3.south) -- (g3.north);
    \draw[decorate,decoration={brace,amplitude=4pt},black!60]
      (-0.18,-1.78) -- (-0.18,0.18);
    \node[font=\small,anchor=east] at (-0.42,-0.8) {$\ket\psi$};
    \draw[decorate,decoration={brace,amplitude=4pt},black!60]
      (4.88,0.18) -- (4.88,-1.78);
    \node[font=\small,anchor=west] at (5.12,-0.8) {$U_a\ket\psi$};
  \end{scope}
  \begin{scope}[yshift=-4.7cm]
    \node[ptitle,anchor=west] at (-0.9,1.95) {(b) coordinated};
    \node[cbit] (aa) at (2.35,1.22) {$a$};
    \foreach \y in {0,-0.8,-1.6}
      \draw[qwire] (0,\y) -- (4.7,\y);
    \node[gate] (h1) at (1.3,0)     {$U_1(a)$};
    \node[gate] (h2) at (2.35,-0.8) {$U_2(a)$};
    \node[gate] (h3) at (3.4,-1.6)  {$U_3(a)$};
    \draw[cwire] ($(aa.south)+(-0.09,0)$) -- (h1.north);
    \draw[cwire] (aa.south) -- (h2.north);
    \draw[cwire] ($(aa.south)+(0.09,0)$) -- (h3.north);
    \draw[decorate,decoration={brace,amplitude=4pt},black!60]
      (-0.18,-1.78) -- (-0.18,0.18);
    \node[font=\small,anchor=east] at (-0.42,-0.8) {$\ket\psi$};
    \draw[decorate,decoration={brace,amplitude=4pt},black!60]
      (4.88,0.18) -- (4.88,-1.78);
    \node[font=\small,anchor=west] at (5.12,-0.8) {$U_a\ket\psi$};
  \end{scope}
\end{tikzpicture}
\caption{\label{fig:encoder-coordination}%
Encoder coordination (\cref{def:encoder-structures}), for $n=3$.
Double lines carry the classical label; single lines are the qubits of
$\ket\psi$.  (a)~In an independent encoding, party $j$ selects its
unitary using only its private bit $a_j$; these encoders reach exactly
the LME states (\cref{thm:kraus}).  (b)~In a coordinated encoding, each
local unitary may depend on the entire label $a$.  Coordination concerns
only how the classical label is used: every encoder remains a product
unitary, and no entangling operation is applied.}
\end{figure}

%% file: tables/parallel_characterizations.tex
\begin{table*}[t]
\caption{\label{tab:parallel-characterizations}
Parallel characterizations of the independent and commuting-reflection
classes.  A state belongs to either column exactly when it satisfies any, and
hence all, of that column's three entries.  The ancilla row on the left is the
standard LME definition; on the right it is the coherent-control
reformulation discussed below.}
\begin{ruledtabular}
\renewcommand{\arraystretch}{1.35}
\begin{tabular}{@{}p{0.16\textwidth}p{0.36\textwidth}p{0.40\textwidth}@{}}
\textbf{Viewpoint}
&
\textbf{Independent class $\EI=\mathrm{LME}$}
&
\textbf{Commuting-reflection class $\ECI$ (this work)}
\\
\colrule
\textbf{Flatness}
&
Up to local unitaries, the state is flat in a product basis.
&
Up to local unitaries, the state is flat in a graph-state basis.
\\
\textbf{Encoding}
&
There exist $n$ independent local binary choices,
$U_a=\bigotimes_j U_j(a_j)$, whose orbit is an orthonormal basis.
&
There are $2^n$ product reflections that commute pairwise and whose orbit
is an orthonormal basis.
\\
\textbf{Coherent ancillas}
&
Party-local controlled gates, each coupling $A_j$ only to $S_j$,
maximally entangle the ancilla register with the system (locally maximally
entangleable).
&
The same holds with the label register controlling the commuting
reflections jointly, but the control need not factor across the pairs
$A_jS_j$.
\end{tabular}
\end{ruledtabular}
\end{table*}

%% file: figures/m4_graph_classes.tex
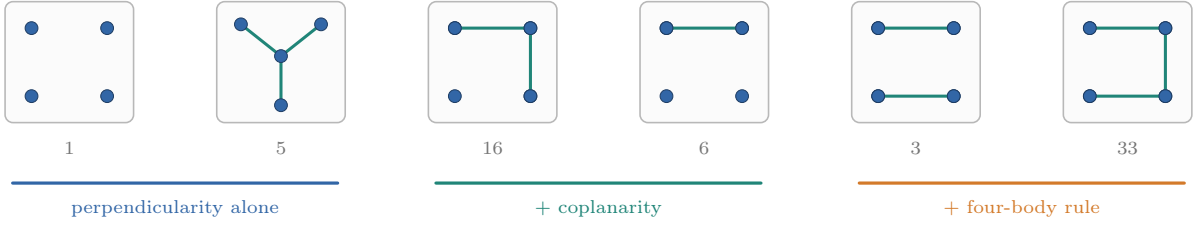
\begin{figure*}[t]
\centering
\begin{tikzpicture}[x=1cm,y=1cm,
  vert/.style={circle,fill=encoderblue,draw=encoderblue!60!black,
               line width=0.3pt,inner sep=1.7pt},
  gedge/.style={encoderteal,line width=1.15pt,line cap=round},
  gpanel/.style={rounded corners=3pt,draw=black!28,fill=black!1.5,
                 line width=0.6pt}]
  \foreach \s in {0,...,5}{
    \begin{scope}[xshift=\s*2.8cm]
      \draw[gpanel] (0,-0.05) rectangle (1.7,1.55);
    \end{scope}
  }
  \begin{scope}
    \node[vert] at (0.35,1.2) {};
    \node[vert] at (1.35,1.2) {};
    \node[vert] at (1.35,0.3) {};
    \node[vert] at (0.35,0.3) {};
    \node[font=\scriptsize,text=black!55] at (0.85,-0.38) {$1$};
  \end{scope}
  \begin{scope}[xshift=2.8cm]
    \coordinate (w1) at (0.85,0.83);
    \coordinate (w2) at (0.32,1.25);
    \coordinate (w3) at (1.38,1.25);
    \coordinate (w4) at (0.85,0.18);
    \draw[gedge] (w1) -- (w2);
    \draw[gedge] (w1) -- (w3);
    \draw[gedge] (w1) -- (w4);
    \node[vert] at (w1) {};
    \node[vert] at (w2) {};
    \node[vert] at (w3) {};
    \node[vert] at (w4) {};
    \node[font=\scriptsize,text=black!55] at (0.85,-0.38) {$5$};
  \end{scope}
  \begin{scope}[xshift=5.6cm]
    \node[vert] (v1) at (0.35,1.2) {};
    \node[vert] (v2) at (1.35,1.2) {};
    \node[vert] (v3) at (1.35,0.3) {};
    \node[vert] (v4) at (0.35,0.3) {};
    \draw[gedge] (v1.center) -- (v2.center);
    \draw[gedge] (v2.center) -- (v3.center);
    \node[vert] at (0.35,1.2) {};
    \node[vert] at (1.35,1.2) {};
    \node[vert] at (1.35,0.3) {};
    \node[font=\scriptsize,text=black!55] at (0.85,-0.38) {$16$};
  \end{scope}
  \begin{scope}[xshift=8.4cm]
    \node[vert] (v1) at (0.35,1.2) {};
    \node[vert] (v2) at (1.35,1.2) {};
    \node[vert] (v3) at (1.35,0.3) {};
    \node[vert] (v4) at (0.35,0.3) {};
    \draw[gedge] (v1.center) -- (v2.center);
    \node[vert] at (0.35,1.2) {};
    \node[vert] at (1.35,1.2) {};
    \node[font=\scriptsize,text=black!55] at (0.85,-0.38) {$6$};
  \end{scope}
  \begin{scope}[xshift=11.2cm]
    \node[vert] (v1) at (0.35,1.2) {};
    \node[vert] (v2) at (1.35,1.2) {};
    \node[vert] (v3) at (1.35,0.3) {};
    \node[vert] (v4) at (0.35,0.3) {};
    \draw[gedge] (v1.center) -- (v2.center);
    \draw[gedge] (v3.center) -- (v4.center);
    \node[vert] at (0.35,1.2) {};
    \node[vert] at (1.35,1.2) {};
    \node[vert] at (1.35,0.3) {};
    \node[vert] at (0.35,0.3) {};
    \node[font=\scriptsize,text=black!55] at (0.85,-0.38) {$3$};
  \end{scope}
  \begin{scope}[xshift=14cm]
    \node[vert] (v1) at (0.35,1.2) {};
    \node[vert] (v2) at (1.35,1.2) {};
    \node[vert] (v3) at (1.35,0.3) {};
    \node[vert] (v4) at (0.35,0.3) {};
    \draw[gedge] (v1.center) -- (v2.center);
    \draw[gedge] (v2.center) -- (v3.center);
    \draw[gedge] (v3.center) -- (v4.center);
    \node[vert] at (0.35,1.2) {};
    \node[vert] at (1.35,1.2) {};
    \node[vert] at (1.35,0.3) {};
    \node[vert] at (0.35,0.3) {};
    \node[font=\scriptsize,text=black!55] at (0.85,-0.38) {$33$};
  \end{scope}
  \draw[encoderblue,line width=1.2pt,line cap=round]
    (0.1,-0.85) -- (4.4,-0.85);
  \node[font=\scriptsize,text=encoderblue] at (2.25,-1.18)
    {perpendicularity alone};
  \draw[encoderteal,line width=1.2pt,line cap=round]
    (5.7,-0.85) -- (10.0,-0.85);
  \node[font=\scriptsize,text=encoderteal] at (7.85,-1.18)
    {$+$ coplanarity};
  \draw[encoderorange,line width=1.2pt,line cap=round]
    (11.3,-0.85) -- (15.6,-0.85);
  \node[font=\scriptsize,text=encoderorange] at (13.45,-1.18)
    {$+$ four-body rule};
\end{tikzpicture}
\caption{\label{fig:m4-graph-classes}%
Why $\ket{M_4}$ admits no commuting-reflection encoding
(\cref{prop:m4-separator}).  Up to local complementation and
relabelling, the $64$ labelled four-vertex graphs fall into six classes
(representatives shown; numbers give the class sizes).  Flatness in a
graph basis would force every nonidentity graph stabilizer to have zero
expectation in $\ket{M_4}$, which translates into elementary geometry:
vanishing two-qubit correlators force perpendicular Bloch axes,
vanishing three-qubit correlators force coplanar axes, and a four-body
rule handles the remaining stabilizers (\cref{app:w3-m4-proofs}).  The
bars indicate the strongest ingredient needed to exclude each pair of
classes; since all six classes are excluded,
$\ket{M_4}\notin\ECI^{(4)}$.}
\end{figure*}

%% file: figures/no_go_overview.tex
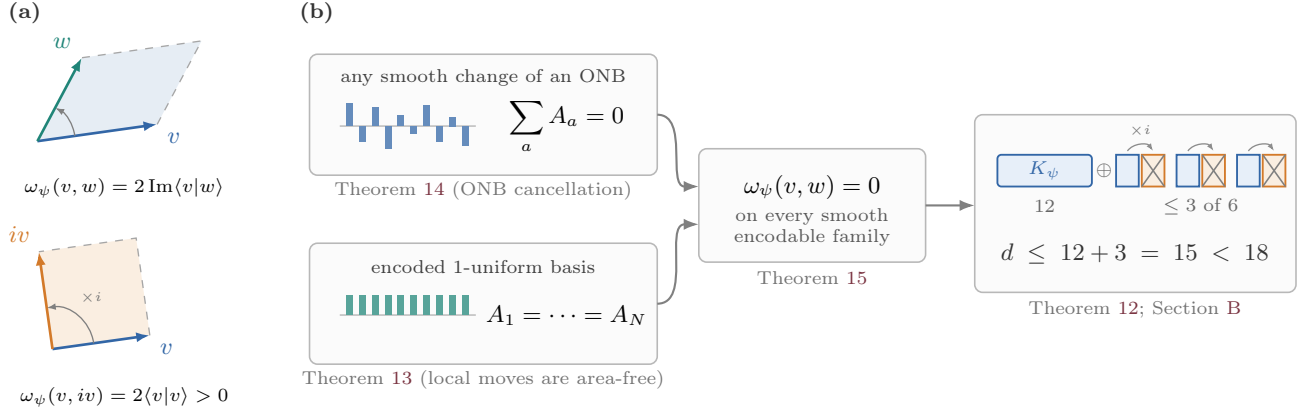
\begin{figure*}[t]
\centering
\begin{tikzpicture}[x=1cm,y=1cm,
  gpanel/.style={rounded corners=3pt,draw=black!28,fill=black!1.5,
                 line width=0.6pt},
  ptitle/.style={font=\footnotesize\bfseries,text=black!80},
  tag/.style={font=\scriptsize,text=black!55},
  boxhead/.style={font=\scriptsize,text=black!75},
  flow/.style={-{Latex[length=2mm]},black!50,line width=0.8pt}]

  \node[ptitle,anchor=west] at (0.05,5.3) {(a)};
  \begin{scope}[shift={(0.55,3.6)}]
    \coordinate (vv) at (8:1.6);
    \coordinate (ww) at (62:1.25);
    \fill[encoderblue!12] (0,0) -- (vv) -- ($(vv)+(ww)$) -- (ww) -- cycle;
    \draw[dashed,black!35,line width=0.5pt]
      (vv) -- ($(vv)+(ww)$) -- (ww);
    \draw[-{Latex[length=1.8mm]},encoderblue,line width=1pt]
      (0,0) -- (vv);
    \draw[-{Latex[length=1.8mm]},encoderteal,line width=1pt]
      (0,0) -- (ww);
    \draw[-{Latex[length=1.2mm]},black!50,line width=0.5pt]
      (8:0.5) arc[start angle=8,end angle=62,radius=0.5];
    \node[font=\small,text=encoderblue,anchor=north west]
      at ($(vv)+(0.02,0.0)$) {$v$};
    \node[font=\small,text=encoderteal,anchor=south east]
      at ($(ww)+(-0.02,0.02)$) {$w$};
  \end{scope}
  \node[font=\scriptsize] at (1.7,2.98)
    {$\omega_\psi(v,w)=2\operatorname{Im}\langle v|w\rangle$};
  \begin{scope}[shift={(0.75,0.85)}]
    \coordinate (vb) at (8:1.3);
    \coordinate (ivb) at (98:1.3);
    \fill[softorange] (0,0) -- (vb) -- ($(vb)+(ivb)$) -- (ivb) -- cycle;
    \draw[dashed,black!35,line width=0.5pt]
      (vb) -- ($(vb)+(ivb)$) -- (ivb);
    \draw[-{Latex[length=1.8mm]},encoderblue,line width=1pt]
      (0,0) -- (vb);
    \draw[-{Latex[length=1.8mm]},encoderorange,line width=1pt]
      (0,0) -- (ivb);
    \draw[-{Latex[length=1.2mm]},black!50,line width=0.5pt]
      (8:0.55) arc[start angle=8,end angle=98,radius=0.55];
    \node[font=\tiny,text=black!55] at (53:0.85) {$\times i$};
    \node[font=\small,text=encoderblue,anchor=north west]
      at ($(vb)+(0.02,0.0)$) {$v$};
    \node[font=\small,text=encoderorange,anchor=south east]
      at ($(ivb)+(0,0.02)$) {$iv$};
  \end{scope}
  \node[font=\scriptsize] at (1.7,0.22)
    {$\omega_\psi(v,iv)=2\langle v|v\rangle>0$};

  \node[ptitle,anchor=west] at (3.9,5.3) {(b)};

  \draw[gpanel] (4.15,3.15) rectangle (8.75,4.75);
  \node[boxhead] at (6.45,4.44) {any smooth change of an ONB};
  \draw[black!30,line width=0.5pt] (4.55,3.8) -- (6.35,3.8);
  \foreach \k/\h in {0/0.30,1/-0.20,2/0.24,3/-0.30,4/0.14,
                     5/-0.10,6/0.27,7/-0.21,8/0.12,9/-0.26}{
    \fill[encoderblue!75]
      ({4.68+0.17*\k-0.045},3.8) rectangle ({4.68+0.17*\k+0.045},{3.8+\h});
  }
  \node[font=\small] at (7.55,3.8) {$\displaystyle\sum_a A_a=0$};
  \node[tag] at (6.45,2.95) {\cref{lem:onb-cancellation} (ONB cancellation)};

  \draw[gpanel] (4.15,0.65) rectangle (8.75,2.25);
  \node[boxhead] at (6.45,1.94) {encoded 1-uniform basis};
  \draw[black!30,line width=0.5pt] (4.55,1.3) -- (6.35,1.3);
  \foreach \k in {0,...,9}{
    \fill[encoderteal!75]
      ({4.68+0.17*\k-0.045},1.3) rectangle ({4.68+0.17*\k+0.045},1.56);
  }
  \node[font=\small] at (7.55,1.3) {$A_1=\cdots=A_N$};
  \node[tag] at (6.45,0.45) {\cref{lem:local-null} (local moves are area-free)};

  \draw[flow] (8.75,3.95) to[out=0,in=180] (9.3,3.0);
  \draw[flow] (8.75,1.45) to[out=0,in=180] (9.3,2.5);
  \draw[gpanel] (9.3,2.0) rectangle (12.3,3.5);
  \node[font=\small] at (10.8,3.0) {$\omega_\psi(v,w)=0$};
  \node[font=\scriptsize,text=black!65] at (10.8,2.6)
    {on every smooth};
  \node[font=\scriptsize,text=black!65] at (10.8,2.33)
    {encodable family};
  \node[tag] at (10.8,1.8) {\cref{prop:equal-contributions}};

  \draw[flow] (12.3,2.75) -- (12.95,2.75);
  \draw[gpanel] (12.95,1.6) rectangle (17.2,3.95);
  \draw[draw=encoderblue,fill=softblue,line width=0.7pt,
        rounded corners=1.5pt] (13.25,3.0) rectangle (14.5,3.42);
  \node[font=\scriptsize,text=encoderblue] at (13.875,3.21) {$K_\psi$};
  \node[font=\scriptsize,text=black!60] at (14.66,3.21) {$\oplus$};
  \foreach \k in {0,1,2}{
    \draw[draw=encoderblue,fill=softblue,line width=0.7pt]
      ({14.82+0.8*\k},3.0) rectangle ({15.12+0.8*\k},3.42);
    \draw[draw=encoderorange,fill=softorange,line width=0.7pt]
      ({15.16+0.8*\k},3.0) rectangle ({15.46+0.8*\k},3.42);
    \draw[black!50,line width=0.6pt]
      ({15.16+0.8*\k},3.0) -- ({15.46+0.8*\k},3.42);
    \draw[black!50,line width=0.6pt]
      ({15.16+0.8*\k},3.42) -- ({15.46+0.8*\k},3.0);
    \draw[-{Latex[length=1mm]},black!45,line width=0.5pt]
      ({14.97+0.8*\k},3.47) to[bend left=55] ({15.31+0.8*\k},3.47);
  }
  \node[font=\tiny,text=black!55] at (15.14,3.74) {$\times i$};
  \node[font=\scriptsize,text=black!60] at (13.875,2.74) {$12$};
  \node[font=\scriptsize,text=black!60] at (15.94,2.74)
    {$\le3$ of $6$};
  \node[font=\small,text=black!80] at (15.07,2.12)
    {$d\;\le\;12+3\;=\;15\;<\;18$};
  \node[tag] at (15.07,1.38)
    {\cref{thm:nonencodable}; \cref{app:technical-dimension}};
\end{tikzpicture}
\caption{\label{fig:no-go-overview}%
The mechanism behind \cref{thm:nonencodable}.  (a)~Two first-order
changes $v,w$ of a state span the signed area $\omega_\psi(v,w)$ (top);
every nonzero direction pairs with $iv$ at strictly positive area
(bottom).  (b)~Structure of the proof.  When any orthonormal basis
varies smoothly, the areas $A_a$ contributed by its $N$ elements cancel
(\cref{lem:onb-cancellation}); when a locally encoded basis of a
1-uniform state varies, the local adjustments contribute no area, so all
elements contribute the same area $A_a=\omega_\psi(v,w)$
(\cref{lem:local-null}).  Equal contributions with zero sum vanish:
$\omega_\psi(v,w)=0$ along every smooth encodable family
(\cref{prop:equal-contributions}).  Among the $18$ tangent directions of
the regular 1-uniform family, the $12$ local rotations $K_\psi$
contribute no area, while the remaining six form three pairs $(v,iv)$;
an area-free family can keep at most one member of each pair, so at most
$12+3=15<18$ directions coexist (\cref{sec:dimension-count};
\cref{app:technical-dimension} extends the bound to the full encodable
set).}
\end{figure*}